\pdfoutput=1
\RequirePackage{ifpdf}
\ifpdf 
\documentclass[pdftex]{sigma}
\else
\documentclass{sigma}
\fi

\numberwithin{equation}{section}

\newtheorem{Theorem}{Theorem}[section]
\newtheorem{Corollary}[Theorem]{Corollary}
\newtheorem{Lemma}[Theorem]{Lemma}
\newtheorem{Proposition}[Theorem]{Proposition}
 { \theoremstyle{definition}
\newtheorem{Definition}[Theorem]{Definition}
\newtheorem{Example}[Theorem]{Example}
\newtheorem{Remark}[Theorem]{Remark} }

\newcommand{\Z}{\mathbb{Z}}
\newcommand{\C}{\mathbb{C}}
\newcommand{\ip}{{\tilde{p}}}
\newcommand{\is}{{\tilde{s}}}

\usepackage{tikz-cd}
\usetikzlibrary{matrix,arrows}

\begin{document}

\newcommand{\arXivNumber}{2007.01241}

\renewcommand{\thefootnote}{}

\renewcommand{\PaperNumber}{143}

\FirstPageHeading

\ShortArticleName{Riemannian Geometry of a Discretized Circle and Torus}
\ArticleName{Riemannian Geometry of a Discretized Circle\\ and Torus\footnote{This paper is a~contribution to the Special Issue on Noncommutative Manifolds and their Symmetries in honour of~Giovanni Landi. The full collection is available at \href{https://www.emis.de/journals/SIGMA/Landi.html}{https://www.emis.de/journals/SIGMA/Landi.html}}}

\Author{Arkadiusz BOCHNIAK, Andrzej SITARZ and Pawe{\l}~ZALECKI}
\AuthorNameForHeading{A.~Bochniak, A.~Sitarz and P. Zalecki}
\Address{Institute of Theoretical Physics, Jagiellonian University, \\
prof.\ Stanis\l awa \L ojasiewicza 11, 30-348 Krak\'ow, Poland}
\Email{\href{mailto:arkadiusz.bochniak@doctoral.uj.edu.pl}{arkadiusz.bochniak@doctoral.uj.edu.pl}, \href{mailto:andrzej.sitarz@uj.edu.pl}{andrzej.sitarz@uj.edu.pl}, \newline
\hspace*{14mm}\href{mailto:pawel.zalecki@doctoral.uj.edu.pl}{pawel.zalecki@doctoral.uj.edu.pl}}

\ArticleDates{Received July 03, 2020, in final form December 15, 2020; Published online December 23, 2020}

\Abstract{We extend the results of Riemannian geometry over finite groups and provide a~full classification of all linear connections for the minimal noncommutative differential calculus over a finite cyclic group. We solve the torsion-free and metric compatibility condition in general and show that there are several classes of solutions, out of which only special ones are compatible with a metric that gives a Hilbert $C^\ast$-module structure on the space of the one-forms. We compute curvature and scalar curvature for these metrics and find their continuous limits.}


\Classification{46L87; 83C65}

\begin{flushright}
\begin{minipage}{55mm}
\it Dedicated to Giovanni Landi\\ on the occasion of 60th birthday
\end{minipage}
\end{flushright}

\renewcommand{\thefootnote}{\arabic{footnote}}
\setcounter{footnote}{0}

\section{Introduction}

Noncommutative geometry \cite{BM_book,Connes:1994bk, La:bk} offers new insight into spaces and their generalizations
by allowing to skip the traditional assumptions of points and to use the methods of differential geometry that are
adaptable to the algebraic language. The construction of differential calculi has been one of the first steps that allowed
the extension of the formalism of gauge theory to the realm of noncommutative spaces. In particular, the spaces that consist
of a~finite number of points or discrete lattices have appeared not only as discrete approximations of differentiable spaces
but as manifolds in the generalized sense \cite{Balachandran:1996, BHDS, DMH}.

One of the crucial aspects of differential geometry is, however, the link between
the metric aspects, that is, distances and the norm on the space of states with the
relevant objects in the differential algebra. In the classical differential geometry this link
is provided by the metric tensor and leads to the notion of metric compatible and
torsion-free linear connections that provide relevant and physically significant constructions
of the curvature and appropriate geometrical objects like Ricci tensor and the scalar curvature.
The noncommutative geometry has been, so far, unable to retrace these steps
in full generality despite many efforts. Apart from the usual problem of the choice of
the differential calculus for the given algebra, the main problem is the definition of
the metric over the bimodule of differential forms and the linear connection \cite{MMM,mourad}.
The choice of the metric and the linear connection that are compatible with the bimodule structure
of the differential forms lead necessarily to severe restrictions not only on the possible
metrics~\cite{Si94} but also on connections~\cite{DM94} and curvature~\cite{DMM96}.

Recently an updated version of the approach to linear connections for a~special type of
differential calculi was studied in general and for particular examples of noncommutative
spaces \cite{BGJ, BGL, BGL2, BGM}.

A systematic approach to the general issue of bimodule linear connections and bimodule
metrics over the differential forms was started by Majid~\cite{BM11,majid2} and developed
in~\cite{BM14}. The formalism used there allows to generalize in a systematic way almost
all classical notions like torsion-freeness and Ricci and scalar curvature~\cite{BM_book}
(depending on certain choices).

It is worth mentioning that the Connes' spectral approach based on the Dirac operators~\cite{Connes:1994bk}
that was much successful in the reconstruction of the Einstein--Hilbert gravity action for the standard
and modified almost-commutative geometries \cite{CoMa:bk08} can be applied to the studies of the
generalized scalar of curvature for certain noncommutative manifolds, in particular the noncommutative
tori (see \cite{FK13,FK19} and \cite{DaSi15} for the specific example of an asymmetric torus).
 Recently progress in recovering some geometric objects like Ricci tensor was made
\cite{FGK19}, yet there is currently no method to recover all such objects, e.g.,
the torsion, through the spectral methods.
Nevertheless, there exists a huge discrepancy between the usual methods of recovering the
geometric objects like the scalar of curvature for the manifolds and their deformations and the
attempt to use of spectral methods \cite{Ba15,BG19,Gl17,PS} in the finite-dimensional case.

In this paper we start systematic computations of linear connections for finite groups, starting with the
easiest example of finite cyclic groups and their products.

This extends the known results to metrics, which are more general than the so-called quantum symmetric
or edge symmetric already explored in~\cite{majid2020,M13}. We provide a complete classification of linear
connections that are torsion-free and compatible with any nondegenerate metric, demonstrating that there are severe
restrictions on possible metrics and only a certain class of them allows the existence of non-unique compatible
linear connections. The existence of large number of possible solutions (i.e., the non-triviality of the moduli
space of Levi-Civita connection) without imposing additional conditions was already observed in~\cite{M1} for the
case $\mathbb{Z}_2\times\mathbb{Z}_2$ (see also \cite[Chapter~8.2]{BM_book}). Some of the arguments we use here
were already present in~\cite{M1,M2}.

The main result of this paper is that for the special case of left-right symmetric metric there still exist a choice
of linear connections that are torsion-free and compatible with the metric, a scalar curvature that vanishes
only for the constant (equivariant) metric (for some canonical choices of the arbitrary constants in the theory).
We demonstrate that the freedom is much larger in the case of the products of two discrete circles even in the
case of the constant metric.

\section{Preliminaries}

We start by recalling basic definitions and known results about differential calculi over finite groups.

In what follows for a group $G$ with its neutral element $e$ we denote
$G^\times = G \setminus \{e\}$. For a~subset $H \subset G^\times$ by $H^\perp$ we shall always
denote $G^\times \setminus H$. Furthermore, for $X \subset G$ we denote by~$\chi_X$
the characteristic function of the set $X$, i.e.,
\[
\chi_X(g)=\begin{cases}1, & g\in X, \\
0, & g\not\in X.\end{cases}
\]

\begin{Definition}The (first-order) differential calculus over an algebra $A$ over a field $k$ is a pair $\big(\Omega^1(A), d\big)$, where
$\Omega^1(A)$ is a bimodule over $A$, ${\rm d}$ is a linear map ${\rm d}\colon A\rightarrow \Omega^1(A)$, which satisfies the Leibniz rule,
\begin{equation*}
{\rm d}(ab)=a\, {\rm d}b+({\rm d}a) b,
\end{equation*}
and $\Omega^1(A)$ is generated as a left module by the image of ${\rm d}$. We say
that $\big(\Omega^1(A),{\rm d}\big)$ is connected if $\ker {\rm d} \cong k$.
\end{Definition}
In case the algebra $A$ is a Hopf algebra (with a coproduct $\Delta$, counit $\varepsilon$ and antipodal map~$S$) we have the following.
\begin{Definition}
We say that the differential calculus over a Hopf algebra $A$ is left-covariant if there exists a coassociative left coaction of
$A$ on $\Omega^1(A)$, $\delta_{L}\colon \Omega^1(A) \rightarrow A \otimes \Omega^1(A)$, such that
\begin{equation*}
\delta_{L}(a\omega b)=\Delta(a)\delta_L(\omega)\Delta(b),
\end{equation*}
for every $a,b\in A$, $\omega\in\Omega^1(A)$, and
\begin{equation*}
\delta_{L}\circ {\rm d} =(\mathrm{id}\otimes {\rm d})\circ \Delta.
\end{equation*}
\end{Definition}

In a similar manner we define the right-covariance and bicovariance (as simultaneous left and right covariance). The canonical example of a first-order differential calculus is given by the universal calculus, with $\Omega_u^1(A) = \ker m\subseteq A\otimes A$, where $m\colon A\otimes A\rightarrow A$ is the multiplication map for $A$, and the universal differential ${\rm d}_u\colon A \rightarrow \ker m$ of the form ${\rm d}_ua:=a\otimes 1 -1\otimes a$. The universal calculus over a Hopf algebra is bicovariant. The bicovariant calculi over an arbitrary Hopf algebra were classified by Woronowicz~\cite{Wo89}.

For a $\ast$-algebra $A$ we can consider differential calculi that in addition possess a $\ast$-structure, that is we assume ${\rm d}$ is a derivation of a $\ast$-algebra, i.e., ${\rm d}(a^\ast)=({\rm d}a)^\ast$ for every $a\in A$.

\subsection{Finite cyclic groups}

Geometric aspects of finite groups have been intensively studied by several authors, including \cite{BHDS,DMH}, and also
recently in \cite{M1} for the infinite cyclic group $\mathbb{Z}$ and~\cite{M2} in case of the
group $\mathbb{Z}_2\times\mathbb{Z}_2$.
As we consider finite cyclic groups $\mathbb{Z}_N$ with $N\ge 2$ many results are much simpler and therefore we skip the derivation of them, which are mostly adaptations of well-known ones published in the aforementioned literature.

We denote by $e_g$, $g\in G$ a function that vanishes everywhere apart from $g$, $e_g(h) =\delta_{g,h}$.
By~$R_g(f)$ we denote the right translation $R_g(f)(h) = f(hg).$ In a similar manner we introduce the left translation~$L_g$. Note that the left covariance of the calculus is equivalent to the fact that for all $\omega\in \Omega^1(A)$ and
all $g\in G$ we have $L_g\omega\in \Omega^1(A)$.

 The following theorem summarizes the properties of calculi on finite cyclic groups, which we will intensively use in the rest
of the paper. Since these results are well-known, we only sketch the proof.

\begin{Theorem}
Let $G=\Z_N$. Each connected, star-compatible first-order
bicovariant differential calculus over $C(\Z_N)$ is determined by a subset
$H \subset G^\times$ such that $H = H^{-1}$ and $H$ generates the entire group $\Z_N$. By $|H|$ we denote the number of elements in $H$. There are $|H|$ left invariant forms
\[ \theta_h=\sum\limits_{g\in G}{\rm d}e_{gh}\, e_g, \qquad h \in H, \]
such that the star- and the bimodule structure over $\Omega^1(C(\Z_N))$ is
\[ (\theta_h)^* = - \theta_{h^{-1}}, \qquad f \theta_g = \theta_g R_g(f), \]
and the calculus is inner
\[ {\rm d} f = [\theta, f], \qquad \theta = -\sum\limits_{h \in H} \theta_h. \]
Moreover, there exists a unique minimal extension of the first order differential calculus
$($as defined by Woronowicz$)$ so that
\[ \theta_g \wedge \theta_h = - \theta_h \wedge \theta_g, \qquad h,g \in H. \]
\end{Theorem}
\begin{proof}
The first part follows directly from~\cite[Section~2]{DMH}. From~\cite[Proposition 3.1]{Wo89}, for any bicovariant differential calculus $\big(\Omega^1_H(\mathbb{Z}_N),{\rm d}\big)$ there exists a unique bimodule automorphism~$\sigma_W$ of $\Omega^1_H(\mathbb{Z}_N)\otimes_{C(\mathbb{Z}_N)} \Omega^1_H(\mathbb{Z}_N)$ given by
\[ \sigma_W(\theta_g\otimes\theta_h)=\theta_{g^{-1}hg}\otimes\theta_h.\]
Then we introduce the symmetrization map as the extension of the $\sigma_W$ to the tensor algebra of~$\Omega^1(A)$. Its
kernel is identified with the exterior algebra over~$A$.
\end{proof}

\begin{Theorem} \label{dcZN}
Consider the cyclic group $\Z_N$, $N>2$ with the generator~$p$. Denote by~$\tilde{p}$ its inverse in $\Z_N$.
Then there exists a minimal bicovariant, star-compatible connected differential calculus, generated by
$\theta_p$, $\theta_\ip$ with the following structure:
\begin{alignat*}{3}
& f \theta_p = \theta_p R_p(f), \qquad && f \theta_\ip = \theta_\ip R_\ip(f), &\nonumber\\
& \theta_p^* = - \theta_\ip, \qquad && {\rm d} f = - [ \theta_p + \theta_\ip, f ], &\nonumber\\
& {\rm d} \theta_p = 0, \qquad && {\rm d} \theta_\ip = 0,& \nonumber\\
& {\rm d} \omega = \theta \wedge \omega + \omega \wedge \theta.\qquad & && 
\end{alignat*}
\end{Theorem}
\begin{proof}Since $N\neq 2$ we have $p\neq \ip$. The first order differential calculus generated by $H=\{p,\ip\}$ is then bicovariant,
connected and compatible with the star structure.

By a direct computation we see that for any $g\in H$:
\[ {\rm d}\theta_g=\sum\limits_{a\in H}\chi_H\big(ga^{-1}\big)\theta_a\wedge \theta_{ga^{-1}}.\]
Since $H=\{p,\ip\}$ and if $p\in H$, then for $N>3$ we get $p\ip^{-1}\not\in H$ (otherwise $p=e$ or is order $3$), we immediately infer that ${\rm d}\theta_p={\rm d}\theta_{\ip}=0$. For $N=3$ in the above sum there is only one term: ${\rm d}\theta_g=\theta_{g^{-1}}\wedge \theta_{g^{-1}}$ for $g=p,\ip$, and as a result ${\rm d}\theta_p={\rm d}\theta_\ip=0$ also for this case.

Notice that since ${\rm d}\theta_p={\rm d}\theta_{\ip}=0$, for $\omega=\omega_p\theta_p+\omega_{\ip}\theta_{\ip}$ we get
\[ {\rm d}\omega={\rm d}\omega_p\wedge\theta_p+{\rm d}\omega_{\ip}\wedge\theta_{\ip}.\]
Furthermore,
\begin{gather*}
 {\rm d}\omega_p=-\omega_p\theta-R_{p}(\omega_p)\theta_{\ip}-R_{\ip}(\omega_p)\theta_p,\\
 {\rm d}\omega_\ip=-\omega_\ip\theta-R_{p}(\omega_\ip)\theta_{\ip}-R_{\ip}(\omega_\ip)\theta_p.
\end{gather*}
As a result
\[ {\rm d}\omega =-\big[(\omega_p-\omega_{\ip})\theta_p\wedge\theta_\ip+\theta_\ip\wedge\big(\omega_p-R_{\ip^2}\omega_\ip\big)\theta_p\big],\]
which is exactly $ \theta \wedge \omega + \omega \wedge \theta$.
\end{proof}

\section{Bimodule linear connections}
Following \cite{BM14}, for a first-order differential calculus $\big(\Omega^1(A),{\rm d}\big)$, we set
\begin{Definition}\label{lincon}
A linear connection on the bimodule $\Omega^1(A)$ is a pair $(\nabla, \sigma)$, a linear map
\begin{equation*}
\nabla\colon \ \Omega^1(A)\rightarrow \Omega^1(A)\otimes_A\Omega^1(A),
\end{equation*}
and a bimodule map,
\begin{equation*}
\sigma\colon \ \Omega^1(A)\otimes_A\Omega^1(A)\rightarrow \Omega^1(A)\otimes_A \Omega^1(A),
\end{equation*}
called the generalized braiding, such that
\begin{gather*}
\nabla(a\omega)={\rm d}a\otimes_A\omega +a\nabla\omega, \nonumber\\
\nabla(\omega a)=(\nabla \omega)a +\sigma (\omega\otimes_A {\rm d}a),
\end{gather*}
for all $a\in A$, $\omega\in \Omega^1(A)$.
\end{Definition}
Notice that (see, e.g., \cite[Proposition~2.1.3]{BM11}) with such a definition, the linear connection can be naturally extended to the whole tensor algebra $T^\bullet\Omega_H^1(A)$ with $T^0\Omega_H^1(A):=A$, $\nabla|_{A}:={\rm d}$ and $\sigma|_{A\otimes_A \Omega_H^1(A)}=\mathrm{id}$.

In our case of the algebra $C(\Z_N)$ since the calculus is inner, we can use
\begin{Remark}
If there exists $\theta\in \Omega^1(A)$ such that ${\rm d}a=[\theta, a]$, then any bimodule connection is of the form
\begin{equation}
\nabla\omega=\theta\otimes\omega-\sigma(\omega\otimes\theta)+\alpha(\omega),
\label{inner_connection}
\end{equation}
for some bimodule maps $\sigma$, $\alpha$ \cite[Theorem~2.1]{M13}.
\end{Remark}
As an immediate consequence of the above definition we get the following result:
\begin{Proposition}\label{connection_determined}
For a minimal bicovariant calculus over $C(\Z_N)$ with $N\neq 3$ a~bimodule linear connection is determined by a bimodule map~$\sigma$.
\end{Proposition}
\begin{proof}

The argument is exactly the same as in \cite{M1,M2} for groups $\mathbb{Z}_2\times\mathbb{Z}_2$ and~$\mathbb{Z}$, respectively.

We use shortened notation $\Omega^1(A)$ to denote $\Omega^1_H(\C(\Z_N))$ from Theorem~\ref{dcZN}.
First of all, observe that there are no bimodule maps apart from the zero map between
and $\Omega^1(A)$ and $\Omega^1(A) \otimes_A \Omega^1(A)$. Indeed, there are no objects in
$\Omega^1(A) \otimes_A \Omega^1(A)$ that have the same bimodule commutation rules as in
$\Omega^1(A)$. Otherwise $p$ would be of order $3$.
Therefore, necessarily $\alpha\equiv0$. Hence, the bimodule connection $\nabla$ and $\sigma$ are mutually determined.
\end{proof}
\subsection{Torsion-free connection}

Let us now concentrate on the notion of a torsion. We define torsion as a map $T_\nabla\colon \Omega^1_H(A)\rightarrow \Omega^2_H(A)$
given by
\begin{equation*}
T_\nabla=\wedge \circ \nabla-{\rm d}.
\end{equation*}

Following \cite{M13} we say that the connection is compatible with a torsion if $\operatorname{Im}(\mathrm{id}+\sigma)\subseteq \ker\wedge$.
The connection is said to be torsion-free if $T_\nabla=0$. Observe first that we have the following result.
\begin{Proposition}
For a minimal bicovariant calculus over $C(\Z_N)$ with the torsion-free connection, the map $\sigma$ must satisfy,
\begin{equation}
\omega \wedge \theta= - {\wedge} \circ \, \sigma(\omega \otimes_A \theta).
\label{sigmawedge}
\end{equation}
\end{Proposition}
\begin{proof}
Comparing
\begin{gather*}
\nabla(\omega f)=\nabla(\omega)f+\sigma(\omega\otimes_A {\rm d}f),\\
{\rm d}(\omega f)={\rm d}\omega f-\omega\wedge {\rm d}f,
\end{gather*}
we immediately get $ \omega \wedge {\rm d} f= - {\wedge}\circ\, \sigma(\omega \otimes_A {\rm d}f)$, which gives us the claimed formula.
\end{proof}

Notice that it follows from the last proposition that the torsion-free connection is compatible with a torsion. This is a manifestation of the more general result:
\begin{Proposition} For an inner calculus~\eqref{inner_connection} with the extension to $\Omega^2(A)$ such
that $\theta\wedge\theta=0$, the connection is torsion-free if and only if is torsion-compatible and $\operatorname{im} \alpha\subseteq\ker\wedge$.
\end{Proposition}
\begin{Remark}
Notice that the similar result was stated in \cite[Theorem~2.1(3)]{M13} but in that formulation $\alpha$ was forced to be a zero bimodule map instead of satisfying $\operatorname{im} \alpha\subseteq\ker\wedge$.
This was later corrected in \cite[Proposition~8.11]{BM_book}. As one can easily see in the example of the~$\mathbb{Z}_3$ group with
a generator $p$, there exists a torsion-free connection with nontrivial~$\alpha$, namely $\alpha(\theta_p)=\theta_{p^2}\otimes_A\theta_{p^2}$ (because $p^2 p^2 =p$).
\end{Remark}

From Proposition~\ref{connection_determined} it follows that for $N\neq 3$ the pair $(\nabla, \sigma)$ is mutually unambiguously determined. The case with $N=3$ has to be considered separately. Even the torsion-freeness does not guarantee vanishing of $\alpha$.

\begin{Definition}[{see \cite[p.~572]{BM_book}}]
We say that the connection is star-compatible, if
\[ \nabla\circ \ast=\sigma\circ \dagger \circ \nabla,\]
where $(\omega\otimes_A \eta)^\dagger=\eta^\ast\otimes_A\omega^\ast$, i.e.,~$\dagger$ is the induced $\ast$-structure on higher tensors.
\end{Definition}

As an immediate consequence we obtain the following two results, which were already formulated in~\cite{majid2020}
(see also~\cite{M2} for the case $G=\mathbb{Z}$).

\begin{Proposition}\label{sigmy}The torsion-free bimodule connections over the minimal bicovariant calculi over $C(\Z_N)$ with $N\neq 4$ are determined by a family of functions $A_p$, $A_\ip$, $B_p$, $B_\ip$, so that~$\sigma$ is
\begin{gather*}
 \sigma(\theta_p \otimes_A \theta_p) = A_p \theta_p \otimes_A \theta_p,\\
\sigma(\theta_\ip \otimes_A \theta_\ip) = A_\ip \theta_\ip \otimes_A \theta_\ip, \\
\sigma(\theta_p \otimes_A \theta_\ip) = B_p ( \theta_p \otimes_A \theta_\ip + \theta_\ip \otimes_A \theta_p ) - \theta_p \otimes_A \theta_\ip, \\
\sigma(\theta_\ip \otimes_A \theta_p) = B_\ip ( \theta_p \otimes_A \theta_\ip + \theta_\ip \otimes_A \theta_p ) - \theta_\ip \otimes_A \theta_p.
\end{gather*}
\end{Proposition}
\begin{proof}
If follows directly from the fact that $\sigma$ is a bimodule map, $p^2\neq \ip^2$ for $N\neq 4$, and the compatibility condition
of $\sigma$ with the $\wedge$ \eqref{sigmawedge}.
\end{proof}
The assumption for the connection to be compatible with the star structure imposes further restrictions on the functions $A$ and $B$.
\begin{Proposition}
The connection in the Proposition~{\rm \ref{sigmy}} is star-compatible if and only if the relations below are fulfilled:
\begin{gather}
\big(R_{g}\overline{A_g}\big)\big(R_{g^{-1}}A_{g^{-1}}\big)=1,\qquad
|B_g-1|^2 + \overline{B_g}B_{g^{-1}}=1\label{star-connection}
\end{gather}
for $g\in\{p,\ip\}$.
\end{Proposition}

\section{Metric}
We use here the notion of metricity as introduced in \cite{BM14} (see also \cite{Si94}),
\begin{Definition}Let $\big(\Omega^1_H(A),{\rm d}\big)$ be a first order differential calculus over~$A$. We define the metric as a pair, an element
\[
\mathbf{g}=\mathbf{g}^{(1)}\otimes_A \mathbf{g}^{(2)}\in \Omega^1_H(A)\otimes_A\Omega^1_H(A),
\]
and a bimodule map
\[
(\cdot,\cdot)\colon \ \Omega^1_H(A)\otimes_A\Omega^1_H(A)\rightarrow A,
\] such that the pairing
between them is nondegenerate, in the following sense,
\begin{equation*}
\big(\omega,\mathbf{g}^{(1)}\big)\mathbf{g}^{(2)}=\omega=\mathbf{g}^{(1)}\big(\mathbf{g}^{(2)},\omega\big)
\end{equation*}
for all $\omega\in\Omega^1_H(A)$.
\end{Definition}

\begin{Definition}\label{def:comp1}
We say that the metric $\mathbf{g}$ is compatible with $\ast$, if $\mathbf{g}^*=\mathbf{g}$, that is
\begin{gather*}
\mathbf{g}^\ast = \big( \mathbf{g}^{(1)}\otimes_A \mathbf{g}^{(2)} \big)^*
= \big(\mathbf{g}^{(2)}\big)^*\otimes_A \big(\mathbf{g}^{(1)}\big)^* = \mathbf{g}, \\
(\omega^*, \rho^*) = (\rho, \omega)^*, \qquad \forall\, \omega,\rho \in \Omega^1(A).
\end{gather*}
\end{Definition}
\begin{Definition}\label{def:comp2}
We say that the metric is compatible with the higher-order differential calculus iff
$\mathbf{g} \in \ker \wedge$, that is
\[ {\wedge}\, \mathbf{g} = \mathbf{g}^{(1)}\wedge \mathbf{g}^{(2)} = 0. \]
\end{Definition}

\begin{Remark}The metric satisfying condition from Definition \ref{def:comp2} is called {\it quantum symmetric} in \cite{BM14}, while the condition in Definition \ref{def:comp1} is called {\it reality} therein.
\end{Remark}

In our situation, we have:
\begin{Lemma}
\label{lemma_metric01}
A nondegenerate metric over the minimal bicovariant calculus over $C(\Z_N)$
is given by functions $G_p$, $G_\ip$, which are everywhere different from $0$,
\begin{gather}
 \mathbf{g} = G_p \theta_p \otimes_A \theta_\ip + G_\ip \theta_\ip \otimes_A \theta_p,\label{metric1} \\
 ( \theta_a,\theta_b)=\frac{1}{R_{a^{-1}} G_{a^{-1}}} \delta_{a^{-1},b} , \qquad a,b=\{p,\ip\}.\nonumber 
\end{gather}
\end{Lemma}

\begin{Corollary}\label{cor_metric}
If the metric in the Lemma~{\rm \ref{lemma_metric01}} is also compatible with the higher-order dif\-fe\-rential calculus $($i.e., $\wedge{\bf g}=0)$, then it can be described by the only one function $G:=G_p=G_\ip$.
\end{Corollary}

\begin{proof}Although this result is well-known (for example, the case with $G=\mathbb{Z}$ was proved in \cite{M2} where also the
consequences for the metric compatible with higher-order calculi were studied), for completeness we demonstrate the proof.

Since
\begin{equation*}
\theta_g f=\big(R_{g^{-1}} f\big) \theta_g,
\end{equation*}
for arbitrary $f \in \C(G)$, then we can now analyse the conditions we have from the required properties
of a metric $\mathbf{g}$. First, we obviously have
\[ f(\rho,\omega)=(\rho,\omega)f, \qquad (\rho f,\omega)=(\rho,f\omega),\]
and
\[ f(\rho,\omega)=(f\rho,\omega), \qquad (\rho,\omega f)=(\rho,\omega)f, \]
for every $f\in\C(G)$ and every $\rho,\omega\in \Omega^1_H(A)$. Therefore, we have
\begin{align*}
f(\theta_i,\theta_j)&=(\theta_i,\theta_j)f=(\theta_i,\theta_j f)=\big(\theta_i,\big(R_{j^{-1}}f\big)\theta_j\big) \\
&=\big(\theta_i\big(R_{j^{-1}}f\big),\theta_j\big)
=\big(\big(R_{i^{-1}}R_{j^{-1}}f\big) \theta_i,\theta_j\big) = \big(R_{i^{-1}}R_{j^{-1}}f\big)(\theta_i,\theta_j).
\end{align*}
Since the right action is free it implies that $j=i^{-1}$ whenever $(\theta_i,\theta_j)\neq0$. Therefore
the bimodule map $(\cdot, \cdot)$ has to be of the following form
\begin{equation*}
(\theta_a,\theta_b)=\delta_{a^{-1},b} F_{a},
\end{equation*}
where $F_a \in \C(\Z_N)$.
We are now ready to explore conditions that follow from equation~\eqref{metric1}. Let us write $\mathbf{g}$ in the basis, here $H=\{p,\ip\}$,
\begin{equation*}
\mathbf{g} = \sum\limits_{a,b \in \{p,\ip\}} \mathbf{g}_{ab} \theta_a \otimes_A \theta_b,
\end{equation*}
and consider the condition $\omega=\mathbf{g}^{(1)}\big(\mathbf{g}^{(2)},\omega\big)$ with
$\omega = \theta_c$. We have
\begin{gather*}
\theta_c
=\sum\limits_{a,b \in H} \mathbf{g}_{ab} \theta_a (\theta_b, \theta_c)=
\sum\limits_{a,b \in H} \mathbf{g}_{ab} \theta_a F_b \delta_{b,c^{-1}}
=\sum\limits_{a \in H} \mathbf{g}_{ac^{-1}} \theta_a F_{c^{-1}}.
\end{gather*}
The equality holds if and only if
\begin{equation*}
\mathbf{g}_{ac^{-1}} R_{a^{-1}}\big(F_{c^{-1}}\big) = \delta_{a,c}.
\end{equation*}
Taking $a=c$ we immediately obtain the claimed result.
\end{proof}

It follows immediately that $\mathbf{g}$ is a central element in $\Omega^1_H(A) \otimes_A \Omega^1_H(A)$ and
we can compute both contractions of the metric, that is not only $ \big(\mathbf{g}^{(1)},\mathbf{g}^{(2)}\big)$
but also $ \big(\mathbf{g}^{(2)},\mathbf{g}^{(1)}\big)$ make sense.
We have
\begin{gather}
\big(\mathbf{g}^{(1)},\mathbf{g}^{(2)}\big)=G_p(\theta_p,\theta_{\ip})+G_\ip(\theta_\ip,\theta_p)=\frac{G_p}{R_\ip G_\ip}+\frac{G_\ip}{R_{p} G_p}, \nonumber\\
\big(\mathbf{g}^{(2)},\mathbf{g}^{(1)}\big)=(\theta_\ip,G_p\theta_p)+(\theta_p,G_\ip\theta_{\ip})=(R_pG_p)\frac{1}{R_pG_p}+(R_\ip G_\ip)\frac{1}{R_\ip G_\ip}=2.\label{metric-short}
\end{gather}

\begin{Definition}
The metric is right-invariant if $R_{{h}}(\mathbf{g}) = \mathbf{g}$ (resp.\ left-invariant if $L_{{h}}(\mathbf{g}) = \mathbf{g}$),
for every $h\in G$, where we have used the unique extension of right (resp.\ left) translations to the whole differential algebra, so that
\[ R_g({\rm d}f)={\rm d}(R_gf), \qquad (\mathrm{resp.} \ L_g({\rm d}f)={\rm d}(L_gf)).\]
\end{Definition}

\begin{Lemma}The metric $\mathbf{g}$ is left-invariant if and only if for every $g\in\{p,\ip\}$, $G_p={\rm const}$. A~nondegenerate
metric is $\ast$-compatible iff for the metric coefficients are real, $G_g=G_g^\ast$.
\end{Lemma}

Finally let us see when a $*$-compatible metric defines a norm on the module of one-forms.

\begin{Lemma}\label{Hilbert}
Let us define $($see also {\rm \cite[Proposition~8.40]{BM_book})}:
\[ \langle \cdot, \cdot \rangle \colon \ \Omega^1(A)\otimes_A\Omega^1(A)\rightarrow A,
\qquad \langle \omega_1, \omega_2 \rangle:=(\omega_1^\ast,\omega_2). \]
If all $G_g$ are real and negative then $\Omega^1(A)$ equipped with $\langle \cdot,\cdot\rangle$ is a Hilbert $C^\ast$-module over~$A$.
\end{Lemma}
\begin{proof}
The defined map is sesquilinear (right $\C$-linear, left antilinear) and satisfies
\[ \langle \omega_1 a_1, \omega_2 a_2 \rangle= a_1^* \langle \omega_1,\omega_2\rangle a_2,\]
for every $a_1,a_2 \in A$ and all $\omega_{1},\omega_2\in\Omega^1(A)$.
Furthermore, if ${\bf g}={\bf g}^\ast$ then also $\langle \omega_1, \omega_2\rangle^\ast = \langle \omega_2, \omega_1\rangle$ and $\langle \omega,\omega\rangle \ge 0$ if all $G_g$
are negative-valued. Moreover, in such a case $\langle \omega,\omega\rangle=0$ iff $\omega=0$. To sum up, for ${\bf g}={\bf g}^*$ with negative-valued $G_g$, we indeed
have a pre-Hilbert module structure. Therefore, $\|\omega\|:=\|\langle \omega, \omega\rangle\|^{\frac{1}{2}}$ defines a norm on $\Omega^1(A)$,
making its completion (which in a finite-dimensional case is $\Omega^1(A)$) a
Hilbert $C^\ast$-module over $A$.
\end{proof}

\section{Metric compatibility condition}
Let us now concentrate on the metric compatibility condition for a bimodule linear connection over
the minimal bicovariant calculus on $C(\Z_N)$. Although we shall later concentrate on the solutions that correspond to the real-valued metrics that provide nondegenerate scalar products over $\Omega^1$, we solve the metric compatibility
problem in all generality.

\begin{Definition}[{\cite{BM14}, see also \cite[Chapter 8]{BM_book}}]
A linear connection $(\nabla,\sigma)$ is said to be compatible with the metric ${\bf g}$ if
\begin{equation*}
(\nabla\otimes\mathrm{id}){\bf g}+(\sigma\otimes\mathrm{id})(\mathrm{id}\otimes\nabla){\bf g}=0.
\end{equation*}
\end{Definition}

Before we proceed with the conditions for the general $\Z_N$ case, $N>4$, let us consider a much simpler
case of $N=2$.
\begin{Example}[{Levi-Civita bimodule connections for $\Z_2$, compare \cite[Lemma~2.1]{M1}}]
In the case of $\Z_2$, we have $p=\ip$ and therefore the entire connection is determined by one
function $S$:
\[ \nabla(\theta_p) = (S-1) \theta_p \otimes_A \theta_p, \qquad \sigma(\theta_p \otimes_A \theta_p)= S \theta_p \otimes_A \theta_p,\]
the metric is given by $G \theta_p \otimes_A \theta_p$ and the metric compatibility then reads:
\[ (G-R_pG) + G (S-1) + G S (R_p(S)-1) = 0. \]
Using notation $G_0 = G(e)$, $G_1 = G(p)$ and $S_0$, $S_1$ for the respective values of~$S$ we have
\[ G_0-G_1+G_0(S_0-1)+G_0S_0(S_1-1)=0,\]
and
\[ G_1-G_0 + G_1 (S_1-1) + G_1 S_1 (S_0-1)=0.\]
The above system of equations is equivalent to the following two
\[ G_1=G_0 S_0 S_1, \qquad
G_0=G_1S_1 S_0,\]
which lead to $G_1 = \pm G_0$ and
\[S_0 S_1 =\pm 1.\]
Observe that even in the case of constant metric we can have a one-parameter family of torsion-free,
metric compatible connections given by
\[S_0 = z , \qquad S_1 = \frac{1}{z}. \]
\end{Example}
\begin{Theorem}\label{comp_N}
For the torsion-free bimodule connection for the minimal bicovariant calculus over $\mathbb{Z}_N$ with $N>4$ the metric
compatibility conditions takes the following form:
\begin{gather}
 G_g\big(R_{g^{-1}}B_{g^{-1}}\big)A_g=R_{g^{-1}}G_g,\nonumber\\
G_{g^{-1}}\big(R_gB_g-1\big)B_{g^{-1}}+G_g(B_g-1)\big(R_{g^{-1}}A_{g^{-1}}\big)=0,\nonumber\\
R_gG_g=G_{g^{-1}}(R_gB_g-1)\big(B_{g^{-1}}-1\big)+G_gB_g\big(R_{g^{-1}}A_{g^{-1}}\big),\label{rownfin2}
\end{gather}
for $g=p,\ip$.
\end{Theorem}
\begin{proof}First, notice that
\begin{gather*}
(\nabla\otimes\mathrm{id}){\bf g} =
(G_pA_p - R_\ip G_p)\theta_p\otimes_A\theta_p\otimes_A\theta_\ip +(G_\ip B_\ip-R_\ip G_\ip)\theta_p\otimes_A\theta_\ip\otimes_A\theta_p \\
\hphantom{(\nabla\otimes\mathrm{id}){\bf g} =}{}
+G_\ip(B_\ip-1)\theta_\ip\otimes_A\theta_p\otimes_A\theta_p+G_p(B_p-1)\theta_p\otimes_A\theta_\ip\otimes_A\theta_\ip\\
\hphantom{(\nabla\otimes\mathrm{id}){\bf g} =}{} +(G_\ip A_\ip-R_pG_\ip)\theta_\ip\otimes_A\theta_\ip\otimes_A\theta_p+(G_p B_p-R_pG_p )\theta_\ip\otimes_A\theta_p\otimes_A\theta_\ip.
\end{gather*}
On the other hand
\begin{gather*}
(\sigma\otimes\mathrm{id})(\mathrm{id}\otimes\nabla){\bf g} =
G_p(R_\ip B_\ip -1)A_p\theta_p\otimes_A\theta_p\otimes_A\theta_\ip+G_\ip(R_p B_p -1)A_\ip\theta_\ip\otimes_A\theta_\ip\otimes_A\theta_p \\
\hphantom{(\sigma\otimes\mathrm{id})(\mathrm{id}\otimes\nabla){\bf g} =}{}
+ [G_p(R_\ip B_\ip -1)(B_p-1)+G_\ip B_\ip(R_p A_p -1) ]\theta_p\otimes_A\theta_\ip\otimes_A\theta_p \\
\hphantom{(\sigma\otimes\mathrm{id})(\mathrm{id}\otimes\nabla){\bf g} =}{}
+ [G_\ip(R_p B_p -1)(B_\ip-1)+G_p B_p(R_\ip A_\ip -1) ]\theta_\ip\otimes_A\theta_p\otimes_A\theta_\ip \\
\hphantom{(\sigma\otimes\mathrm{id})(\mathrm{id}\otimes\nabla){\bf g} =}{}
+ [G_p(R_\ip B_\ip-1) B_p+G_\ip (B_\ip-1)(R_p A_p -1) ]\theta_\ip\otimes_A\theta_p\otimes_A\theta_p \\
\hphantom{(\sigma\otimes\mathrm{id})(\mathrm{id}\otimes\nabla){\bf g} =}{}
+ [G_\ip(R_p B_p-1) B_\ip+G_p(B_p-1)(R_\ip A_\ip -1) ]\theta_p\otimes_A\theta_\ip\otimes_A\theta_\ip.
\end{gather*}
Taking the sum of these two expressions we get the final result.
\end{proof}

Let us now solve the system of equations \eqref{rownfin2}. To start we substitute $A_g=a_g+1$ and $B_g=b_g+1$, then the equations
read
\begin{gather}
R_{g^{-1}}G_g=G_g(1+a_g)\big(1+R_{g^{-1}}b_{g^{-1}}\big),\nonumber\\
G_gb_g\big(1+R_{g^{-1}}a_{g^{-1}}\big)+G_{g^{-1}}(R_gb_g)\big(1+b_{g^{-1}}\big)=0,\nonumber\\
R_gG_g=G_g(1+b_g)\big(1+R_{g^{-1}}a_{g^{-1}}\big)+G_{g^{-1}}b_{g^{-1}}(R_gb_g).\label{rownfin2a}
\end{gather}
for $g=p,\ip$.

Introducing $X_g=\frac{R_gG_g}{G_{g^{-1}}}$ and combining the first and the third equation we obtain
\begin{equation}
b_{g^{-1}} ( R_g b_g ) =X_g - R_{g^{-1}} X_g. \label{war_X3_1}
\end{equation}

As the left-hand side is unchanged when we replace $g$ by $g^{-1}$ and act on the result with~$R_g$, we obtain
\begin{equation*}
X_g - R_{g^{-1}} X_g = R_g \big( X_{g^{-1}}- R_{g} X_{g^{-1}} \big),
\end{equation*}
Since $X_g$ satisfies
\[ R_g X_{g^{-1}} = \frac{1}{X_g}, \]
we obtain
\[
X_g - R_{g^{-1}} X_g =\frac{1}{X_g} - R_{g} \frac{1}{X_g},
\]
which leads to
\[
X_g + R_{g} \frac{1}{X_g} = R_{g^{-1}} \left( X_g + R_g \frac{1}{X_g} \right),
\]
and as a result
\begin{equation}
\label{condition_cyclic}
X_g+\frac{1}{R_gX_g}=c=\mathrm{const}.
\end{equation}
Notice that the above relation is, effectively equivalent to $\big({\bf g}^{(1)},{\bf g}^{(2)}\big)=c$, which means that in this case both contractions as
computed in \eqref{metric-short} are constant.

Writing explicitly $X_p$, $X_{\tilde{p}}$ as functions over $\Z_N$, the relation~\eqref{condition_cyclic} can be
reformulated in the form of the following recurrence system, here for simplicity we denote the function~$X_p$
as $f$ and choose $p=1$ (so $\tilde{p}=-1$),
\begin{equation}\label{recur-f}
\begin{cases}
(c-f(n))f(n + 1)=1,\\
f(0)=f(N),
\end{cases}
\end{equation}
for a function $f\colon \mathbb{N} \rightarrow \mathbb{C}$. Note that we can equivalently choose the equation for
$X_{\tilde{p}}$ (denote this function as $F$) but this corresponds to the choice of $-1$ as the generator of $\Z_N$
and therefore give the equations
\begin{equation*}
\begin{cases}
(c-F(n))F(n - 1)=1,\\
F(0)=F(N),
\end{cases}
\end{equation*}
which is equivalent to \eqref{recur-f} since
\[ F(n) = \frac{1}{f(n-1)}.\]

\subsection{Solving the recurrence relation}
We begin with solving the following recurrence equation \eqref{recur-f}. First, let us choose $\gamma$ such that
$(c-\gamma)\gamma=1$. There are two possible solutions of this equation,
\[ \gamma_\pm = \tfrac{1}{2} \big(c \pm \sqrt{c^2-4}\big),\]
which may be, in general, complex numbers and are mutual inverses, that is $\gamma_- = (\gamma_+)^{-1}$.
Fixing one root $\gamma$ we define $f(n)=k(n)+\gamma$, so that the equation we have to solve reduces to
an equivalent one,
\begin{equation*}
k(n)k(n+1)=\frac{1}{\gamma}k(n+1)-\gamma k(n).
\end{equation*}
Since $\gamma\not=0$ then we either have $k\equiv 0$ or all $k(n)$ are different from $0$.
In the first case we have a constant (trivially periodic) solution,
\[ f(n) = \gamma, \]
whereas in the second case we set $h(n)=\frac{1}{k(n)}$ and obtain
\begin{equation*}
h(n+1)=\frac{1}{\gamma^2}h(n)-\frac{1}{\gamma}.
\end{equation*}
The above relation has a solution,
\begin{equation*}
h(n) =
\begin{cases}
\dfrac{\gamma}{\gamma^2-1 } \bigl( H^2 \gamma^{-2n-2} - 1 \bigr) , & \gamma^2 \not=1, \\
H^2 - n \gamma, & \gamma^2 = 1,
\end{cases}
\end{equation*}
where $H^2$ is an arbitrary constant up to the following restrictions:
\begin{gather*}
 \gamma^2 \not=1\colon \ H^2\neq \gamma^{2k+2}, \ k \in \{0,\ldots, N-1\}, \\
 \gamma^2 = 1\colon \ \gamma^{-1} H^2 \notin \{0,1,\ldots,N-1\}.
\end{gather*}
Before we pass to $f$ observe that in the case $\gamma^2=1$ we cannot have a periodic solution for $h$,
since $h(0)=h(N)$ enforces $\gamma=0$, which contradicts our starting point.
If $\gamma^2\not=1$ the periodicity condition is
\[ h(0) = \frac{\gamma}{\gamma^2-1 } \bigl( H^2 \gamma^{-2} - 1 \bigr) = \frac{\gamma}{\gamma^2-1 } \bigl( H^2 \gamma^{-2N-2} - 1 \bigr) = h(N), \]
which is possible only if $\gamma^{2N}=1$ or $H=0$. The solution with $H=0$ is nothing else as a constant solution with
$\gamma^{-1}$ (corresponding to the other choice of the root of the equation $(c-\gamma)\gamma=1$).

We can write explicit form of a non-constant (i.e., with $H\neq 0$) solution for $f$:
\begin{equation}
f(n) = \frac{H \gamma^{-n} - H^{-1} \gamma^{n}}{H\gamma^{-n-1} - H^{-1}\gamma^{n+1}}.\label{rec_f}
\end{equation}

This form of the solution is very convenient, as it is easy to verify the multiplication property for $f$:
\begin{equation*}
\prod\limits_{n=0}^{N-1} f(n) = \frac{H-H^{-1}}{H\gamma^{-N}-H^{-1}\gamma^{N}}=\gamma^N,
\end{equation*}
where we have used $\gamma^{2N}=1$. Note that this holds as well for the constant solution $f(n)=\gamma$.

There are $2N-2$ possible values of $\gamma$ giving non-constant periodic solutions for~$f$, however,
since $c=\gamma+\gamma^{-1}$, both $\gamma$ and $\gamma^{-1}$ result in the same value of~$c$, so that
there are only $N-1$ possible values of~$c$, for which there exist non-constant solutions.
Since $\gamma^{2N}=1$ those $c$ are real.

\subsubsection{The real-valued solutions}
As we are interested in real metrics $G_g$, we consider real-valued solutions of the above recurrence system. It immediately follows from \eqref{rec_f} that
non-constant real solutions exist only for \mbox{$|H|=1$}, i.e.. for $H={\rm e}^{{\rm i}\phi}$ with some $\phi$. Using the fact that $\gamma$ satisfies $\gamma^{2N}=1$, $\gamma^2\neq 1$
and $H^2\neq \gamma^{2n+2}$, $n\in\mathbb{Z}$, we can choose $\gamma = {\rm e}^{\pi {\rm i} \frac{l}{N}}$ and obtain a set of solutions, parametrized by $l=1,\dots ,N-1,\allowbreak N+1,\ldots, 2N-1$,
\begin{equation*}
f_{l,\phi}(n)=\cos\left(\frac{\pi l}{N}\right)+\sin\left(\frac{\pi l}{N}\right)\cot\left(\phi-\frac{\pi l}{N}(n+1)\right).
\end{equation*}
Some of the solutions are, however, repeated as $f_{2N-l,\phi}=f_{l,-\phi}$. Moreover, for such $\gamma$ we have
$c=2\cos\big(\frac{\pi l}{N}\big)$. Note that although we have excluded the case $\gamma^2=1$, the above formula recovers some of
the constant real solutions, which arise for
$l=0$ ($f(n)=1$) and $l=N$ ($f(n)=-1$), so in fact we can extend the range of $l$ also into $l=0$ and $l=N$. It is also easy to see that in case of the real nonconstant solutions~$X_p$ cannot be a positive
function. Finally, let us observe that in case we do not demand reality of the metric, the formula above is still valid but with $\phi$ allowed to be an arbitrary complex number.

\subsubsection{The coefficients of the linear connection}
In the next step we are going to solve the system of equations following from
\eqref{rownfin2a} without restricting ourselves to real solutions of~$X_g$.
Using the first and second equation and~\eqref{war_X3_1} we
obtain a linear dependence between~$b_g$ and~$a_{g^{-1}}$:
\begin{equation*}
G_{g^{-1}} ( R_g b_g + X_g ) = G_g \big( 1 + R_{g^{-1}} a_{g^{-1}} \big) .
\end{equation*}
Reintroducing $1 + R_{g^{-1}} a_{g^{-1}}$ into the first equation we have
\begin{equation}\label{iksy_b}
R_{g^{-1}} X_g = (1+b_g)(X_g +R_g b_g),
\end{equation}
which, after splitting $X_g +R_g b_g$ into $(X_g - 1) + R_g(1 + b_g)$, is equivalent to
\[
R_g(1+b_g)=1 - X_g + \frac{R_{g^{-1}}X_g}{1+b_g}.
\]
Note that $1+b_g$ cannot vanish at any point since $X_g$ cannot vanish at any point, so we can divide both
sides by it. Next, substituting
\[ Y_g = \frac{1+ b_g}{R_{g^{-1}}X_g} +1, \]
we obtain
\[
R_g Y_g = \frac{1}{X_g} \frac{Y_g}{Y_g-1}.
\]
This has an obvious solution $Y_g \equiv 0$, which gives
\[
b_g = - R_{g^{-1}}X_g -1,
\]
and apart from this solution $Y_g$ must be invertible at each point. Then, take
$y_g = (Y_g)^{-1}$ to obtain
\[ R_g y_g =X_{g} (1-y_g ). \]
To solve this equation it is sufficient to find just one solution~$y_g^0$ of the inhomogeneous equation
and a family of solutions of the homogeneous equation
\[
R_g y_g^{{\rm hom}} = - y_g^{{\rm hom}} X_{g}.
\]
The first problem is solved explicitly by verifying that
\[
y_g^0 = \frac{1}{c+2}\big(1+R_{g^{-1}}X_g\big),
\]
provided that $c\neq -2$. We shall discuss the special case $c=-2$ later.

Next, we solve the homogeneous equation. It is easy to see that all solutions are parametrized by a multiplicative constants $\kappa_p$, $\kappa_{\ip}$,
\begin{gather}
y_p^{{\rm hom}}(n) = \kappa_p (-1)^n \prod_{k=0}^{n-1} X_p(k),\label{yhom}\\
y_{\ip}^{{\rm hom}}(n) = \kappa_{\ip} (-1)^n \prod_{k=0}^{n-1} X_p(k) = \frac{\kappa_{\ip}}{\kappa_{p}} y_p^{{\rm hom}}(n)\nonumber
\end{gather}
for $n \in \mathbb{Z}_N$, where $\kappa_g$ are such that $y_g^{\rm hom}+y_g^0\neq 0$, since we require $y_g$ to be invertible.

Observe that for the function $y_g$ to be periodic we need to have
\[\kappa_g (-1)^N \prod_{k=0}^{N-1} X_g(k) = \kappa_g, \]
which, after taking into account that the product of all $X_g(k)$ in the non-constant case is $\gamma^N$ gives
us
\[ \gamma^N = (-1)^N, \qquad \text{or} \qquad \kappa_g=0,\]
further restricting the possible solutions for $X_g$, which then must be parametrised by an integer $l =0,1,\ldots, 2N-1$ such that $N+l$ is always even.
From now on we will always assume that $N+l$ is even, and proceed with the further analysis.

If we have $X_g=\hbox{const}$ then either $\kappa_g=0$ or $X_g^N=(-1)^N$. For real-valued solutions it restricts constant $X_g$ to be $-1$, or, for even $N$, to be $1$. But since here $c \neq -2$ the first possibility is not allowed.

Finally we go back to the case $c=-2$, for which there exists only the constant solution $X_g=-1$. In this case the equation for $b_g$
reduces to
\[
R_g(1+b_g)+\frac{1}{1+b_g}=2,
\]
which, as we already know from the previous subsection, has only one periodic solution $b_g=0$.

To summarize, we have three possible cases:
\begin{itemize}\itemsep=0pt
\item ${X_g = -1}$. In this case $b_g=0$.

\item ${X_g = \text{const}=\gamma}$, ${\gamma \not=-1}$ and ${\gamma^N\neq (-1)^N}$.
In this case the only periodic solutions are constant ones with $b_g=0$ or $b_g=-1-\gamma$, however,
from \eqref{war_X3_1} we see that at least one of $b_g$, $b_{g^{-1}}$ must be $0$, so we have
three possible solutions: $b_g=b_{g^{-1}}=0$, or
$b_g=0$ and $b_{g^{-1}}=-1-\frac{1}{\gamma}$, or $b_g=-1-\gamma$ and $b_{g^{-1}}=0$.

\item ${X_g \not= \text{const}}$ or ${X_g=\gamma}$ with ${\gamma^N=(-1)^N}$ and ${\gamma\neq -1}$.
In this case, combining the results, we have two possibilities
\begin{equation*}
b_g =
\begin{cases}
-1-R_{g^{-1}} X_g, \\
\dfrac{(c+2) R_{g^{-1}}X_g}{1+R_{g^{-1}} X_g+ (c+2) y^{{\rm hom}}_g}-R_{g^{-1}}X_g-1.
\end{cases}
\end{equation*}
where $y^{{\rm hom}}_g$ is expressed in~\eqref{yhom}.
\end{itemize}

Now, what is left in the last case is the compatibility with~\eqref{war_X3_1}. Indeed, although we had determined possible solutions for $b_g$ and
$b_{g^{-1}}$ we must further check whether they are related with each other through~\eqref{war_X3_1}. First observe that if $b_g$ is of the first type,
then from \eqref{war_X3_1} it follows that the solution for $b_{g^{-1}}$ is
\[ b_g=-1-R_{g^{-1}}X_g, \qquad b_{g^{-1}}=-\frac{X_g - c+\frac{1}{X_g}}{1+X_g}. \]
The last expression for $b_{g^{-1}}$ can be rewritten as
\[ \frac{c+2}{X_g+1}-\frac{1+X_g}{X_g}, \]
which is the solution of the second type with the homogeneous part vanishing. Similarly, inserting the solution for $b_g$ of the second type
with $y^{{\rm hom}}_g=0$ to \eqref{war_X3_1}, we end up with the solution for $b_{g^{-1}}$ of the first type.

Our goal is to establish a relation between $y^{{\rm hom}}_{g^{-1}}$ and $y^{{\rm hom}}_g$. We have already discussed cases with vanishing homogeneous parts,
and have shown that they are coupled, in the aforementioned sense, to the solutions of the first kind, so from now on we assume that for both $g$ and $g^{-1}$ we have
a solution for $b$ of the second type and with $y^{{\rm hom}}\neq 0$.

Inserting these two solutions into \eqref{war_X3_1} we end up with
\begin{equation}
y^{{\rm hom}}_{g^{-1}}\big(R_gy^{{\rm hom}}_g\big)=\frac{X_g-R_{g^{-1}}X_g}{(c+2)^2}. \label{rel_y}
\end{equation}

Using \eqref{yhom} we can write \eqref{rel_y} as
\begin{equation*}
\left(\kappa_g (-1)^n \prod_{k=0}^{n-1} X_g(k) \right)
\left( \kappa_{g^{-1}} (-1)^{n+1} \prod_{k=0}^{n} X_g(k) \right)
= \frac{X_g(n)-X_g(n-1)}{(c+2)^2},
\end{equation*}
which gives
\begin{equation*}
(c+2)^2\kappa_g\kappa_{g^{-1}}X_g(0)=\frac{X_g(n-1)-X_g(n)}{\prod\limits_{k=0}^{n-1}X_g(k)X_g(k+1)}.
\end{equation*}
Notice that since $X_g$ satisfies \eqref{condition_cyclic},
we have
\[
\frac{X_g(n-1)-X_g(n)}{X(n-1)X(n)} = \frac{1}{X_g(n)} - \frac{1}{X_g(n-1)}
= X_g(n-2) - X_g(n-1), \]
so the right hand side is independent on $n$, and the equation imposes a
condition on the product of $\kappa_g$ and $\kappa_{g^{-1}}$:
\begin{equation*}
\kappa_g\kappa_{g^{-1}}=\frac{1}{(c+2)^2}\left(\frac{X_g(N-1)}{X_g(0)}-1\right)
=- H^2 \frac{(\gamma-1)^2}{(\gamma+1)^2(H^2-1)^2}.
\end{equation*}

To sum up, we have proven the following result:
\begin{Theorem}\label{theor_ZN}
For the minimal calculus on $\mathbb{Z}_N$, $N>4$ with $H=\{p,\ip\}$ the only allowed torsion-free connections compatible with the metric ${\bf g}$ are
determined by the bimodule map $\sigma$ as in Proposition~{\rm \ref{sigmy}}, where
\[
A_g=\frac{R_{g^{-1}}G_g}{G_g\big(1+R_{g^{-1}}b_{g^{-1}}\big)},
\]
and for $B_g=1+b_g$ we have the following possibilities depending on $X_g$,

{Case I.}
If $X_g\neq {\rm const}$ the only following functions $X_p$ are allowed
\begin{gather}
X_p(n)=\cos\left(\frac{l\pi }{N}\right)+\sin\left(\frac{l\pi}{N}\right)\cot\left(\phi-\frac{(n+1)l\pi}{N}\right),
\label{Xnnonc}
\end{gather}
for $l=1,\dots ,N-1$, and an arbitrary constant $\phi$ such that ${\rm e}^{2{\rm i}\phi}\neq {\rm e}^{\frac{2l\pi}{N}(n+1)}$.

Then with $c=2\cos\big(\frac{l\pi}{N}\big)$ there exist three possible solutions:
\begin{gather*}
(a) \quad b_p=-1-R_{\ip}X_p,\qquad
b_{\ip}=\frac{c+2}{X_p+1}-\frac{1+X_p}{X_p}.
\\ (b) \quad b_{\ip}=-1-R_{p}X_{\ip},\qquad
b_{p}=\frac{c+2}{X_{\ip}+1}-\frac{1+X_{\ip}}{X_{\ip}}.
\end{gather*}
and, provided that $\prod\limits_{k\in\mathbb{Z}_N}X_g(k)=(-1)^N$,
\begin{equation*}
(c) \quad b_g=\frac{(c+2)R_{g^{-1}}X_g}{1+R_{g^{-1}}X_g +(c+2)y_g^{{\rm hom}}}-R_{g^{-1}}X_g-1, \qquad g = \{p,\ip\},
\end{equation*}
where
\[
y_p^{{\rm hom}}(n) = \kappa_p (-1)^n \prod_{k=0}^{n-1} X_p(k),\qquad
y_{\ip}^{{\rm hom}}(n) = \kappa_{p^{-1}} (-1)^n \prod_{k=0}^{n-1} X_p(k).
\]

Furthermore, the constants $\kappa_p$ and $\kappa_{\ip}$ are restricted via a constraint
\[
\kappa_p\kappa_{\ip}=\frac{1}{(c+2)^2}\left(\frac{X_p(N-1)}{X_p(0)}-1\right),
\]
and also requirement that $y_p^{{\rm hom}}+y_p^0\neq 0$.

{Case II.}
If $X_g = \gamma \equiv{\rm const}$:
\begin{itemize}\itemsep=0pt
\item $b_g=b_{g^{-1}}=0$ is always a solution $($independently of $\gamma)$,
\item if $\gamma^N \neq (-1)^N$, then there are two more independent solutions:
\begin{enumerate}\itemsep=0pt
\item[$(a)$] $b_p=0$ and $b_{\ip}=-1-\frac{1}{\gamma}$,
\item[$(b)$] $b_p=-1-\gamma$ and $b_{\ip}=0$.
\end{enumerate}
\item if $\gamma^N=(-1)^N$ and $\gamma\neq -1$ then~\eqref{Xnnonc} is also a solution.
\end{itemize}
Notice that for $\gamma=-1$ the cases $(a)$ and $(b)$ reduce to the first bullet point.
\end{Theorem}

As the next step let us summarize the restrictions on the possible metrics. As we have computed all possible solutions for
\[ X_g = \frac{R_g G_g}{G_{g^{-1}}},\]
so that
\[ G_g(n+1) = G_{g^{-1}}(n) f(n),\]
we can always choose one of the functions $G_p$, $G_{\ip}$ arbitrarily, and then the second one will be determined by the relation above.

\begin{Remark}\label{remark_Hilbert}For the real metric satisfying ${\bf g}={\bf g}^\ast$, the constant solutions above are restricted to real constant $X_g$, whereas the non-constant solutions are restricted by an additional demand that $\phi$ is a real parameter. Only the solutions with $X_g={\rm const} > 0$ give the real metric ${\bf g}$ that equips the module of one-forms with a Hilbert $C^\ast$-module structure (see Lemma~\ref{Hilbert}).
\end{Remark}

\begin{Remark}If we further assume that the metric is compatible with the differential calculus, $\wedge\,{\bf g}=0$, the solution for $X_g$ provides the solution for $G_p=G_\ip$ given by
\[ G(n) = G_0 \prod_{k=0}^{n-1} f_{l,\phi}(k). \]
The only real constant solutions that are compatible with the differential calculus are restricted to $X_g=1$, and, for even $N$, also $-1$, yet only the first one gives a Hilbert $C^\ast$-module structure. Moreover, no non-constant solution gives rise to a Hilbert $C^\ast$-module structure since they are not of constant sign.
\end{Remark}

We can further assume that in addition to compatibility of the metric with the star structure, the connection itself is star-compatible, i.e., relations in \eqref{star-connection} are satisfied.

Using the first relation in \eqref{rownfin2a} we can express $A$ in terms of~$B$, and then the first relation in~\eqref{star-connection} implies that
\begin{equation}\label{star_b_condition}
\overline{B_g}B_{g^{-1}}=\frac{R_{g^{-1}}X_g}{X_g}.
\end{equation}
Observe that since $X_g$ satisfies \eqref{condition_cyclic}, the right-hand side of this equation is non-negative. Indeed, using \eqref{condition_cyclic} we can write
\begin{equation*}
1-\frac{R_{g^{-1}}X_g}{X_g}=(R_{g^{-1}}X_g)^2-c(R_{g^{-1}}X_g)+1,
\end{equation*}
and the problem reduces to examine the quadratic equation $x^2-cx+1=0$, which has no real roots iff $|c|<2$. Hence for those $c$, the right-hand side is always positive. Interestingly, this is the same range of $c$ for which there exist non-constant solutions for $X_g$. On the other hand, for constant solutions $X_g$ combining \eqref{star_b_condition} with the Theorem~\ref{theor_ZN} we see that in these cases $B_g$ has to be equal to $1$. Let us further examine which non-constant solutions determined in Theorem~\ref{theor_ZN} are allowed when compatibility with the star structure is imposed, so we are concentrate on Case I therein. By a straightforward computation we check that cases~(a) and~(b) do not fulfil the condition~\eqref{star_b_condition}. So, suppose now we take solutions as in the case (c) with non-zero homogeneous parts $y_g^{{\rm hom}}$. Using~\ref{condition_cyclic} again, an the fact that $c\neq-2$, the condition~\eqref{star_b_condition} can be reduced to
\[ y_{g^{-1}}^{\rm hom}+\overline{y_g^{\rm hom}}=0.\]
On the other hand, $y^{{\rm hom}}$ satisfy \eqref{rel_y} and $R_gy_g^{\rm hom}=-y_g^{\rm hom}X_g$, so together with the relation above it implies that
\[ (c+2)^2\big|y_g^{\rm hom}\big|^2=1-\frac{R_{g^{-1}}X_g}{X_g},\]
so we get a restriction for possible star-compatible solutions
\[ |B_g-1|=\big|(c+2)y_g^{\rm hom}\big|.\]
Parametrizing
\begin{gather*}
B_g-1=r{\rm e}^{{\rm i}\rho}, \qquad (c+2)y_g^{\rm hom}=r{\rm e}^{{\rm i}\varphi}, \qquad r=\sqrt{1-\frac{R_{g^{-1}}X_g}{X_g}},\\
 a=(c+2)R_{g^{-1}}X_g, \qquad b=1+R_{g^{-1}}X_g,
\end{gather*}
the relation for the solution $B_g$
\[ B_g-1=\frac{(c+2)R_{g^{-1}}X_g}{1+R_{g^{-1}}X_g+(c+2)y_g^{\rm hom}}-\big(1+R_{g^{-1}}X_g\big)\]
can be rephrased as
\[ b^2-a+r^2{\rm e}^{{\rm i}(\varphi+\rho)}+rb\big({\rm e}^{{\rm i}\varphi}+{\rm e}^{{\rm i}\rho}\big)=0.\]
Simple calculations show that $b^2-a=r^2$, so for $r\neq 0$ the star-compatibility condition for a~connection introduces the following constraints on phases $\rho$ and $\varphi$:
\[ {\rm e}^{{\rm i}\rho}=-\frac{r+b{\rm e}^{{\rm i}\varphi}}{b+r{\rm e}^{{\rm i}\varphi}}.\]

As a result we have the following.
\begin{Proposition} Suppose the conditions as specified in the Remark~{\rm \ref{remark_Hilbert}} are satisfied, i.e., we have a Hilbert $C^\ast$-module structure on~$\Omega^1(A)$ given by the metric~$\mathbf{g}$. Then there exists a unique torsion free, metric compatible and star-compatible linear
connection.
\end{Proposition}
\begin{proof}
It follows from the computations before that in such a case we have $X_g={\rm const}>0$ and from the above discussion it follows that the star-compatibility of the connection fixes $B_g$ to be equal to $1$.
\end{proof}

We finish with a remark that this corollary is in a complete agreement with the result obtained in~\cite{majid2020}, where only the case $X_g=1$ was
assumed.

\section{The curvature}
In this section we shall compute the curvature of the torsion-free linear connection compatible
with the metric~$\mathbf{g}$. Though it can be done for arbitrary metrics that satisfy the
compatibility connections, we shall restrict ourselves to the case of real metrics that equip
the bimodule of one-forms with a Hilbert $C^\ast$-module structure. This will restrict
$X_p=\gamma>0$.

\begin{Definition}
The Riemannian curvature for a given connection $\nabla$ is a map
\begin{equation*}
{\bf R}_\nabla\colon \ \Omega^1 \rightarrow \Omega^2 \otimes_A \Omega^1,
\end{equation*}
defined by the following prescription
\begin{equation*}
{\bf R}_\nabla= ({\rm d}\otimes_A\mathrm{id} -\mathrm{id}\wedge \nabla )\nabla.
\end{equation*}
\end{Definition}
By a straightforward computation we get the following:
\begin{Theorem}
The Riemannian curvature for the connection $\nabla$ from Theorem~{\rm \ref{theor_ZN}} is
\begin{equation*}
{\bf R}_\nabla(\theta_g)=\theta_g\wedge\theta_{g^{-1}}\otimes_A \rho_g, \qquad g=p,\tilde{p},
\end{equation*}
where
\begin{gather*}
\rho_g= \big[B_g(R_gA_g)-A_g\big(R_{g^{-1}}B_g\big)-\big(R_{g^{-1}}B_{g^{-1}}-1\big)(B_g-1)\big]\theta_g\\
\hphantom{\rho_g=}{} + \big[\big(R_{g^{-1}}A_{g^{-1}}\big)(1-B_g)+B_g(R_gB_g-1)\big]\theta_{g^{-1}}.
\end{gather*}
\end{Theorem}

To define the objects corresponding to Ricci and scalar curvature we need, however, some more structure.
\begin{Definition}
Let $\iota$ be a bimodule map representing two-forms in $\Omega^1(A) \otimes_A \Omega^1(A)$,
\begin{equation*}
\iota\colon \ \Omega^2\rightarrow \Omega^1\otimes_A\Omega^1,
\end{equation*}
such that the following diagram commutes:
\[
\begin{tikzcd}\Omega^2\arrow[dr,"\mathrm{id}"]\arrow[r,"\iota"]& \Omega^1\otimes_A\Omega^1\arrow[d,"\wedge"] \\ &\Omega^2.
\end{tikzcd}
\]
Then, we define
\begin{equation*}
\widetilde{{\bf R}}_\nabla\equiv(\iota\otimes \mathrm{id}){\bf R}_\nabla \colon \ \Omega^1\rightarrow\Omega^1\otimes_A\Omega^1\otimes_A\Omega^1,
\end{equation*}
and the Ricci tensor is defined as
\begin{equation*}
{\rm Ricci}=\big({\bf g}^{(1)},\widetilde{{\bf R}}_\nabla\big({\bf g}^{(2)}\big)_{(1)}\big)\widetilde{{\bf R}}_\nabla\big({\bf g}^{(2)}\big)_{(2)}\otimes_A \widetilde{{\bf R}}_\nabla\big({\bf g}^{(2)}\big)_{(3)},
\end{equation*}
where the Sweedler's notation on $\Omega^1\otimes_A\Omega^1\otimes_A\Omega^1$ is used.
\end{Definition}

Observe that, the above definition uses the metric unlike the usual definition of the Ricci tensor
that is metric independent and uses the trace.

Following \cite{BM14} we can further define the Einstein tensor and the scalar curvature,
\begin{Definition}
\begin{gather*}
{\rm Einstein}={\rm Ricci}-\frac{\big(\mathrm{Ricci}_{(1)},\mathrm{Ricci}_{(2)}\big)}{\big({\bf g}^{(1)},{\bf g}^{(2)}\big)}{\bf g},\\
R=\big({\bf g}^{(1)},\widetilde{{\bf R}}_\nabla\big({\bf g}^{(2)}\big)_{(1)}\big)
\cdot\big(\widetilde{{\bf R}}_\nabla\big({\bf g}^{(2)}\big)_{(2)}, \widetilde{{\bf R}}_\nabla\big({\bf g}^{(2)}\big)_{(3)}\big).
\end{gather*}
Obviously, $R$ is an element of the algebra $A$.
\end{Definition}

In our case of $A=C(\mathbb{Z}_N)$ with $N>4$ we observe that the most general form of the lifting map $\iota$ is
\begin{equation*}
\iota(\theta_p\wedge \theta_{\ip})=\theta_p\otimes_A\theta_{\ip}+\beta (\theta_p\otimes_A\theta_{\ip}+\theta_{\ip}\otimes_A\theta_p),
\end{equation*}
where $\beta\in C(\mathbb{Z}_N)$. As an immediate consequence we finally obtain for the Ricci tensor,
\begin{equation*}
\mathrm{Ricci}=-\frac{R_{\ip}\beta}{X_{\ip}}\theta_p\otimes_A\rho_{\ip}+\frac{1+R_p\beta}{X_p}\theta_{\ip}\otimes_A\rho_p.
\end{equation*}
Since $\rho_g$ has a form $M_g\theta_g+N_g\theta_{g^{-1}}$ for $g=p,\ip$, we get for the scalar curvature
\begin{equation*}
R=-\frac{1}{X_{\ip}}R_{\ip}\left(\frac{\beta M_{\ip}}{G_{\ip}}\right)+\frac{1}{X_p}R_p\left(\frac{(1+\beta)M_p}{G_p}\right).
\end{equation*}

Since $A_g=\frac{R_{g^{-1}}G_{g}}{G_g(R_{g^{-1}}B_{g^{-1}})}$ we get
\begin{equation*}
M_g=(1-R_{g^{-1}})\left(\frac{B_gG_g}{B_{g^{-1}}(R_gG_g)}\right)+\big(1-R_{g^{-1}}B_{g^{-1}}\big)(B_g-1),
\end{equation*}
and as a result
\begin{gather*}
R= -\frac{1}{X_{\ip}}R_{\ip}\left(\frac{\beta}{G_{\ip}}\right)\left[(R_{\ip}-1)\left(\frac{B_{\ip}G_{\ip}}{B_p(R_{\ip}G_{\ip})}\right) +(1-B_p)(R_{\ip}B_{\ip}-1)\right] \\
\hphantom{R=}{} + \frac{1}{X_{p}}R_{p}\left(\frac{1+\beta}{G_{p}}\right)\left[(R_{p}-1)\left(\frac{B_{p}G_{p}}{B_{\ip}(R_{p}G_{p})}\right) +(1-B_{\ip})(R_{p}B_{p}-1)\right].
\end{gather*}

We can formulate the main theorem.
\begin{Theorem}
\label{positive_curvature}
For a positive parameter $X_p=\gamma>0$ and a metric ${\bf g}$ with $G:=G_p < 0$, for odd~$N$ there exist three possible torsion-free and metric compatible linear connections given by the functions $B_p$, $B_\ip$ with the corresponding Ricci tensor and the scalar curvature $($for an arbitrary lift of $\Omega^2$ given by the function $\beta)$:

Case $(a)$
\begin{gather*}
 B_p=1,\qquad B_\ip=1, \\
{\rm Ricci}(n) = \gamma \beta(n-1) Z_+(n) \theta_p \otimes_A \theta_\ip + \frac{1+\beta(n+1)}{\gamma} Z_-(n+1) \theta_\ip \otimes_A \theta_p, \\
 R(n) = \gamma^2 \beta(n-1) W_+(n) +\frac{1+\beta(n+1)}{\gamma} W_-(n+1).
\end{gather*}

Case $(b)$
\begin{gather*}
B_p=1, \qquad B_\ip=-\frac{1}{\gamma}, \\
{\rm Ricci}(n) = -\beta(n-1)Z_+(n)\theta_p\otimes_A\theta_\ip-(1+\beta(n+1))Z_-(n+1)\theta_\ip\otimes\theta_p\\
\hphantom{{\rm Ricci}(n) =}{} +\beta(n-1)S_-(n)\theta_p\otimes_A\theta_p,\\
 R(n) = -\gamma\beta(n-1)W_+(n)-(1+\beta(n+1))W_-(n+1).
\end{gather*}

Case $(c)$
\begin{gather*}
 B_p=-\gamma, \qquad B_\ip=1, \\
{\rm Ricci}(n) = -\beta(n-1)Z_+(n)\theta_p\otimes_A\theta_{\ip}-(1+\beta(n+1))Z_-(n+1)\theta_\ip\otimes_A\theta_p\\
\hphantom{{\rm Ricci}(n) =}{} -(1+\beta(n+1))S_+(n+1)\theta_\ip\otimes\theta_\ip,\\
R(n) =-\gamma\beta(n-1)W_+(n)-(1+\beta(n+1))W_-(n+1),
\end{gather*}
where
\begin{alignat*}{3}
 &Z_+(n) = \frac{G(n+1)}{G(n)}-\frac{G(n)}{G(n-1)},\qquad && Z_-(n) = \frac{G(n)}{G(n+1)}-\frac{G(n-1)}{G(n)},& \\
 & S_+(n)=\frac{\gamma+1}{\gamma^2}\left(\frac{G(n+1)}{G(n)}-\gamma^2\right),\qquad && S_-(n)=\gamma(\gamma+1)\left(\frac{G(n-1)}{G(n)}-\frac{1}{\gamma^2}\right),&
\end{alignat*}
and
\[
W_\pm(n) = \frac{Z_\pm(n)}{G(n)}.
\]
On the other hand, for even $N$ in addition to the above ones there are also solutions corresponding to the last point in Case II of Theorem~{\rm \ref{theor_ZN}}. In these cases the corresponding Ricci tensor and the scalar curvature are given by
\begin{gather*}
{\rm Ricci}(n) =\beta(n-1)V_+(n)\theta_p\otimes_A\theta_\ip+(1+\beta(n+1))V_-(n)\theta_\ip\otimes_A\theta_p\\
\hphantom{{\rm Ricci}(n) =}{} -\beta(n-1)T_-(n)\theta_p\otimes \theta_p +(1+\beta(n+1))T_+(n)\theta_\ip\otimes_A\theta_\ip, \\
 R(n)=\frac{\beta(n-1)}{G(n)}V_+(n)+\frac{1+\beta(n+1)}{G(n+1)}V_-(n),
\end{gather*}
where
\begin{gather*}
 V_+(n)=\frac{B_\ip(n)}{B_p(n)}\frac{G(n+1)}{G(n)}-\frac{B_p(n)}{B_\ip(n)}\frac{G(n)}{G(n-1)}, \\ V_-(n)=\frac{B_\ip(n)}{B_p(n)}\frac{G(n+1)}{G(n+2)}-\frac{B_p(n)}{B_\ip(n)}\frac{G(n)}{G(n+1)},\\
T_-(n)=(B_\ip(n)-1)\left(\frac{G(n-1)}{G(n)}+\frac{1}{B_\ip(n)}\right), \\ T_+(n)=(B_p(n)-1)\left(\frac{G(n+2)}{G(n+1)}+\frac{1}{B_p(n)}\right).
\end{gather*}
In the above formulas either $B_p=1$ and $B_\ip$ is given by~\eqref{Xnnonc} or the other way around.
\end{Theorem}

\begin{Theorem}
In the case of left-right symmetric metric $\gamma=1$ and the standard choice of the lift $\beta=-\frac{1}{2}$ the scalar curvature is $R(n)=\pm\frac{1}{2}(W_+(n)-W_{-}(n+1))$, i.e.,
\begin{equation}\label{scalaR}
R(n)=\pm\frac{1}{2}\left[\frac{G(n+1)^3+G(n)^3}{G(n+1)^2G(n)^2}-\left(\frac{1}{G(n-1)}+\frac{1}{G(n+2)}\right)\right],
\end{equation}
with the sign $-$ for the case $(a)$ and $+$ for cases $(b)$ and $(c)$. On the other hand, for the special cases discussed at the end of the previous theorem, the scalar curvature is
\[ R(n)=\frac{1}{2} \left[\frac{B_\ip(n)}{B_p(n)}\left(\frac{1}{G(n+2)}-\frac{G(n+1)}{G(n)^2}\right) +\frac{B_p(n)}{B_\ip(n)}\left(\frac{1}{G(n-1)}-\frac{G(n)}{G(n+1)^2}\right)\right].\]
In particular for the constant metric $G$, this curvature vanishes in all these cases.
\end{Theorem}

\begin{Remark}It is interesting to see the continuous limit of the expression \eqref{scalaR}. A simple computation gives that if we denote by $g(t)$ the limit of the $G(n)$ function, for the parametrization of the curve with $t$, then the curvature $R(t)$ becomes
\[ R(t) = \pm\frac{g''(t) g(t) - g'(t)^2}{g(t)^3} = \pm\frac{1}{g(t)} \frac{{\rm d}}{{\rm d}t} \left( \frac{\frac{{\rm d}}{{\rm d}t} g(t)}{g(t)} \right). \]
\end{Remark}

\subsection{Examples of metrics and curvatures}
It is interesting to see how the scalar curvature depends on the metric. Clearly it vanishes for the constant metric, which can be depicted as an equilateral $N$-polygon. On the other hand, if we consider a polygon that approximates the ellipse, that is the respective lengths of the sides correspond to the lengths of lines connecting points on the ellipse like depicted on the Fig.~\ref{fig:el}, we obtain a nontrivial scalar curvature.

\begin{figure}[h!tb]\centering
\includegraphics[height=5cm]{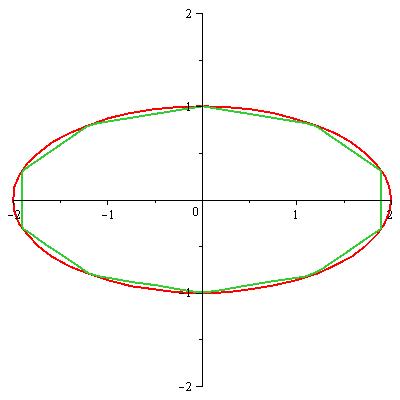}
\caption{Ellipse and ellipse-like polygon ($N=10$).}\label{fig:el}
\end{figure}

We can then compute the scalar curvature for the assumed form of the metric,
which becomes as shown on the Fig.~\ref{fig:emc} and which very closely
approximates its continuous limit.

\begin{figure}[h!tb]\centering\small
\begin{minipage}{50mm}\centering
\includegraphics[height=4cm]{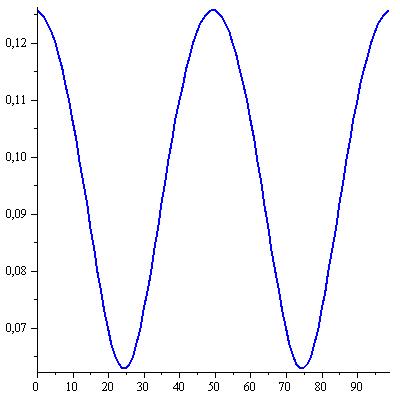}\\
(a)
\end{minipage}\qquad
\begin{minipage}{50mm}\centering
\includegraphics[height=4cm]{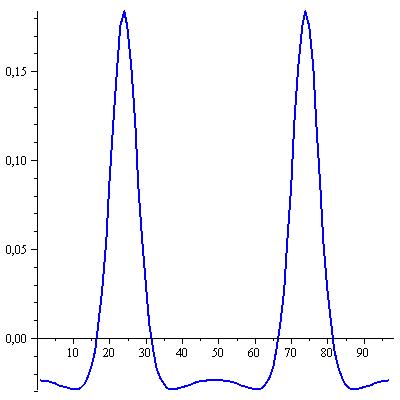}\\
(b)
\end{minipage}
\caption{The metric (a) and the scalar curvature (b) for the ellipse-like polygon ($N=100$).} \label{fig:emc}
\end{figure}

Even more interesting is the inverse problem, of finding the metric such that the scalar curvature is fixed. This shall be treated rather as an exercise in the $N \to \infty$ limit, that is an infinite lattice with the algebra $C(\Z)$, as we fix three distances and then compute the rest using the recursive relation arising
from the Theorem~\ref{positive_curvature}. It is clear that we cannot then guarantee periodic solutions and, moreover, the choice of the initial values
leads to solutions that differ from the continuous approximations.

We have checked some example cases with the constant scalar curvature. It appears that for the positive scalar curvature (see Fig.~\ref{fig:pos}), and the initial data of equilateral sides we obtain oscillating distances, whereas for the negative (small) curvature (see Fig.~\ref{fig:neg}) the metric tends rapidly to zero.

\begin{figure}[h!tb]\centering
\includegraphics[height=4cm]{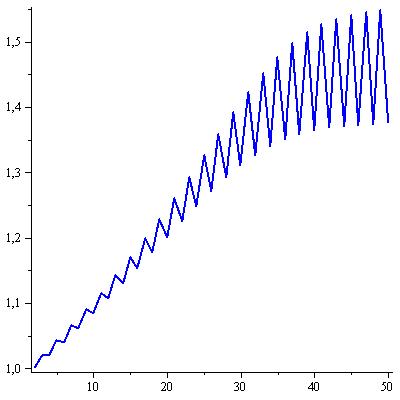} \qquad\qquad
\includegraphics[height=4cm]{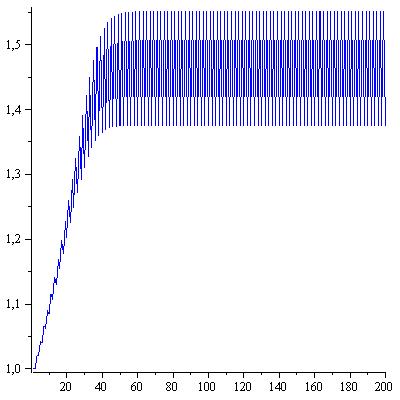}
\caption{The metric for constant positive scalar curvature.}\label{fig:pos}
\end{figure}
\begin{figure}[h!tb]
\centering
\includegraphics[height=4cm]{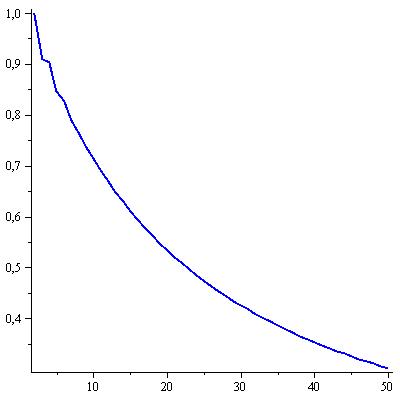} \qquad\qquad
\includegraphics[height=4cm]{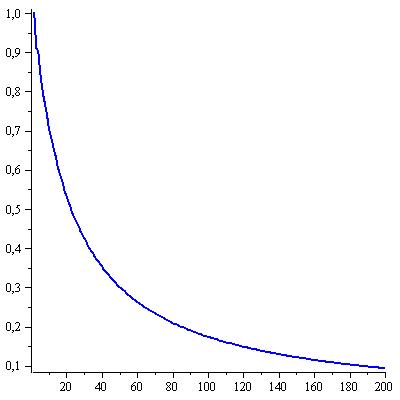}
\caption{The metric for constant negative scalar curvature.}\label{fig:neg}
\end{figure}

\section[Linear connections and curvature for products of ZN]{Linear connections and curvature for products of $\boldsymbol{\Z_N}$}
In this section we shall extend the computations of linear connection to the tensor product of two $C(\Z_N)$
algebras, which corresponds to the Cartesian product of discrete spaces.

Since we consider the minimal differential calculi over both algebras and their
natural graded tensor product, most of the results from previous sections can
be transferred. In particular, it is easy to see that the only possible metric is the
diagonal one, that is for the algebra $A_1 \otimes A_2$,
\[ \mathbf{g} = \mathbf{g}_1 \otimes A_2 + A_1 \otimes \mathbf{g}_2, \]
which means that the total metric is the sum of metrics, yet the coefficients can
be elements of the full algebra.

For simplicity we shall consider here the product of two algebras, this can
be later extended to an arbitrary number of component algebras in the product.
Furthermore, we restrict ourselves only to negative metrics, which then allows
us to use the results of Theorem~\ref{positive_curvature}.

Let us introduce the notation used in this section. We denote by $p$ and $s$ the
generators of the groups~$\Z_N$ and~$\Z_M$, with their inverses~$\ip$ and~$\is$.

\begin{Lemma}The only bimodule metric over $C(\Z_N) \otimes C(\Z_M)$ is of the form
\begin{equation*}
\mathbf{g} = G_p \theta_p \otimes \theta_\ip
+ G_\ip \theta_p \otimes \theta_\ip
+ G_s \theta_s \otimes \theta_\is
+ G_\is \theta_\is \otimes \theta_s,
\end{equation*}
where $G_p$, $G_\ip$, $G_s$, $G_\is$ are functions over $\Z_N \times \Z_M$.
\end{Lemma}

\begin{Lemma}The most general linear connection for the minimal differential calculus over
$C(\Z_N) \otimes C(\Z_M)$ with $N,M>4$ is determined by the map $\sigma$ given by
\begin{gather*}
\sigma(\theta_p \otimes_A \theta_p) = A_p \theta_p \otimes_A \theta_p, \\
\sigma(\theta_\ip \otimes_A \theta_\ip) = A_\ip \theta_\ip \otimes_A \theta_\ip, \\
\sigma(\theta_s \otimes_A \theta_s) = A_s \theta_s \otimes_A \theta_s, \\
\sigma(\theta_\is \otimes_A \theta_\is) = A_\is \theta_\is \otimes_A \theta_\is, \\
\sigma(\theta_p \otimes_A \theta_\ip) = B_p \left( \theta_p \otimes_A \theta_\ip + \theta_\ip \otimes_A \theta_p \right) - \theta_p \otimes_A \theta_\ip + W_{p} (\theta_s\otimes_A\theta_{\is}+\theta_{\is}\otimes_A\theta_s), \\
\sigma(\theta_\ip \otimes_A \theta_p) = B_\ip \left( \theta_p \otimes_A \theta_\ip + \theta_\ip \otimes_A \theta_p \right) - \theta_\ip \otimes_A \theta_p + W_{\ip} (\theta_s\otimes_A\theta_{\is}+\theta_{\is}\otimes_A\theta_s), \\
\sigma(\theta_s \otimes_A \theta_\is) = B_s \left( \theta_s \otimes_A \theta_\is + \theta_\is \otimes_A \theta_s \right) - \theta_s \otimes_A \theta_\is + W_{s} (\theta_p\otimes_A\theta_{\ip}+\theta_{\ip}\otimes_A\theta_p), \\
\sigma(\theta_\is \otimes_A \theta_s) = B_\is \left( \theta_s \otimes_A \theta_\is + \theta_\is \otimes_A \theta_s \right) - \theta_\is \otimes_A \theta_s + W_{\is} (\theta_p\otimes_A\theta_{\ip}+\theta_{\ip}\otimes_A\theta_p),
\end{gather*}
and
\begin{gather*}
 \sigma(\theta_a \otimes_A \theta_b) = C_{ab} (\theta_a\otimes_A\theta_b+\theta_b\otimes_A\theta_a)-\theta_a\otimes_A\theta_b
\end{gather*}
for all $a,\in\{p,\ip \}$, $b \in \{ s,\is\}$ or $ a \in \{ s,\is\}$, $b \in\{p,\ip \}$.
\end{Lemma}
\begin{proof}Since the calculus is inner we can apply \eqref{inner_connection}. Furthermore, in a completely similar manner as in the Proposition~\ref{connection_determined} we infer that if both~$N$ and~$M$ are different than~$3$, $\alpha$ has to be a zero map. As a~result, the connection is determined by the bimodule map~$\sigma$ only. Now, form the bimodule structure, as in the proof of Proposition~\ref{sigmy}, we get the exact form of this map, provided that $N,M\neq 4$.
\end{proof}

\begin{Lemma}The metric compatibility condition, which can be written in general as
\begin{equation*}
\sum\limits_{g,h,k}G_g \psi^{a,b}_{g,k}\big(R_{g^{-1}}\psi^{k,c}_{g^{-1},h}\big)=R_{a^{-1}}G_{c^{-1}}\delta_{b,c^{-1}},
\end{equation*}
for all $a$, $b$, $c$, where $\sigma(\theta_g\otimes_A\theta_h)=\sum\limits_{a,b}\psi_{g,h}^{a,b}\theta_a\otimes_A \theta_b$ leads
to the following system of $36$ equations which can be divided into six types written
explicitly below $($where we use the convention that $h\neq g,g^{-1})$:
\begin{gather*}
 R_{g^{-1}}G_g=G_gA_g\big(R_{g^{-1}}B_{g^{-1}}\big),\\
R_{g^{-1}}G_{g^{-1}}=G_g(B_g-1)\big(R_{g^{-1}}B_{g^{-1}}-1\big)+G_{g^{-1}}B_{g^{-1}}(R_gA_g) \\
\hphantom{R_{g^{-1}}G_{g^{-1}}=}{}
+G_hW_hR_{h^{-1}}\big(C_{h^{-1}g}-1\big)+G_{h^{-1}}W_{h^{-1}}(R_hC_{hg}-1),\\
 0=G_g(B_g-1)\big(R_{g^{-1}}A_{g^{-1}}\big)+G_{g^{-1}}B_{g^{-1}}(R_gB_g-1)\\
\hphantom{0=}{} +G_hW_h\big(R_{h^{-1}}C_{h^{-1}g^{-1}}-1\big)+G_{h^{-1}}W_{h^{-1}}\big(R_hC_{hg^{-1}}-1\big),\\
 0=G_g(B_g-1)\big(R_{g^{-1}}C_{g^{-1}h}-1\big)+G_{g^{-1}}B_{g^{-1}}(R_gC_{gh}-1) \\
\hphantom{0=}{} +G_hW_h\big(R_{h^{-1}}B_{h^{-1}}-1\big)+G_{h^{-1}}W_{h^{-1}}(R_hA_h),\\
0=G_g(C_{gh}-1)\big(R_{g^{-1}}C_{g^{-1}h}\big)+G_hC_{hg}\big(R_{h^{-1}}W_{h^{-1}}\big),\\
R_{g^{-1}}G_h=G_g(C_{gh}-1)\big(R_{g^{-1}}W_{g^{-1}}\big)+G_hC_{hg}\big(R_{h^{-1}}C_{h^{-1}g}\big).
\end{gather*}
\end{Lemma}

Some simplification can arise from considering torsion-freeness together with vanishing of the cotorsion, $\operatorname{co}T_\nabla=({\rm d}\otimes\mathrm{id}-\mathrm{id}\wedge \nabla){\bf g}$,
which is implied~\cite{TM20} by torsion-freeness together with metric compatibility. These
conditions are much simpler since they are linear and in principle can lead to significant restrictions on possible solutions of the main problem.

In our case the cotorsion-freeness can be written explicitly as a system of $16$ equations for functions $A$, $B$, $C$ and $W$:
\begin{gather*}
R_gG_g+G_{g^{-1}}(R_gB_g-1)=G_gR_{g^{-1}}A_{g^{-1}},\\
G_g\big(R_{g^{-1}}C_{g^{-1}h}-1\big)=G_{g^{-1}}(R_gC_{gh}-1),\\
R_{h^{-1}}G_g+G_h\big(R_{h^{-1}}W_{h^{-1}}\big)=G_g\big(R_{g^{-1}}C_{g^{-1}h}\big),
\end{gather*}
where $h\neq g,g^{-1}$ and these indices are taken from $\{p,\ip,s,\is\}$.

Observe that first of them can be used to express $A$ in terms of $B$, and the last one to determine $C$ as a function of $W$. The second one is a compatibility condition for functions~$C$.

It appears, however, that even using these linear dependencies the resulting set
of nonlinear equations is at present beyond the possibilities of exact analytical
solutions. Instead we shall concentrate on showing few possible solutions for
the metrics and compatible linear connections, in particular we want to answer
the question whether for the constant metrics there exist only one compatible
linear connection.

\subsection{Special cases of linear connection for the torus}

\subsubsection[The case with W=0, C=1]{The case with $\boldsymbol{W=0}$, $\boldsymbol{C=1}$}
We begin with considering a special case with all $W$ being zero, which then enforces all $C$ to be
identically $1$ and as a consequence the situation {\it almost} splits into the two parts related with the
two algebras for discrete circles $\mathbb{Z}_N$ and $\mathbb{Z}_M$.
Indeed, for $W=0$ first three relation from the list for the metric compatibility condition reduce to separate equations for groups $\mathbb{Z}_N$ and $\mathbb{Z}_M$, whose solutions we have already found in the previous section. Furthermore, cotorsion-freeness implies that for such a case we have
\[ R_{g^{-1}}C_{g^{-1}h}=\frac{R_{h^{-1}}G_g}{G_g}\]
for all $g$ and $h\neq g,g^{-1}$. Using this result in the fifth condition for metric compatibility we immediately infer that all functions $C$ have to be constantly equal $1$. As a result $R_{h}G_g=G_g$ for all $g$ and $h$ such that $h\neq g,g^{-1}$. We say that in this case the metric is {\it perpendicularly constant}. The remaining relations are automatically fulfilled. Furthermore, for $W=0$ and $C=1$ the connection $\nabla\theta_g=\sum_{a,b}\Gamma^g_{a,b}\theta_a\otimes_A\theta_b$ contains only terms $\Gamma^g_{a,b}$ with $a,b,g\in\{p,\ip\}$ or $a,b,g\in\{s,\is\}$ separately.

These $\Gamma$ functions need to be as determined in Theorem~\ref{theor_ZN} leaving us with the freedom of choosing different solutions. So, although the connection splits into two parts corresponding to two groups~$\mathbb Z_N$ and~$\mathbb Z_M$ both parts can have coefficients depending on both variables due to this freedom. If we enforce that the solutions are constant along the perpendicular direction, using, for example, star compatibility, we indeed have the product geometry and in this case the Riemannian curvature splits into the sum of Riemannian curvatures for two discrete circles. Otherwise, the dependence of each connection on both variables (which can be quite ad hoc) generates mixed terms in Riemannian and Ricci tensors.

\subsubsection{The case of the constant metric}
The previous example shows the existence of nontrivial solutions in the case of the product geometry yet does not show that the solutions are unique. Therefore, for the second example we shall ask the question of all linear connections compatible with the constant metric. Suppose now that the metric coefficients satisfy $G_p=G_\ip = G_s=G_\is$ and are constant. By symmetry arguments
we also assume that all $A$, $B$, $C$ and $W$ are also constant, which are identical (separately) for all $A,B$ and $W$s (as there is symmetry in the change of the space), moreover we assume that $C_{ps}=C_{sp}=C_{\ip\is}=C_{\is\ip}=C_1$ and $C_{p\is}=C_{\ip s}=C_{\is p}=C_{s\ip}=C_2$. The resulting system of equation is then
\begin{gather*}
BA-1 =0, \\
C_1 W + BA + B^2 -2B +C_2W - 2W =0, \\
BA -A + C_1W - 2W -B +B^2 +C_2W = 0,\\
C_1W - C_2 +C_1 C_2 =0, \\
C_2W + C_1C_2 -1 -W =0, \\
AW - 2B +BW -C_2-W+C_1B+C_2B+1=0,
\end{gather*}
and indeed has a unique solution
\[ A=1, \qquad C_1=C_2=1, \qquad W=0, \qquad B=1. \]

We infer from that at least in the case of the constant metric (which is the same for each of the components
of the torus) there exists a unique metric compatible linear connection with certain symmetries. The more general
case, with arbitrary constant coefficients leads to a huge number of nonlinear equations for 20 variables, which
is difficult to solve. Therefore, the only method to proceed is step by step.

To see how this study is involved let us consider the most general case, with the assumption that all $W$s are different
from~$0$. We still assume that $G_p=G_\ip = G_s=G_\is$ are constant, likewise all $C$ and $W$ and suppose now all $W$ are
non-zero. Moreover, we do not impose $B$ to be constant here. The fifth relation in metric compatibility can be now written in the form
\begin{equation*}
 2+W_{h^{-1}}+W_{g^{-1}}=0,
\end{equation*}
where the cotorsion-freeness was used in a completely similar manner as we did it in the previous cases. Changing $g$ into $g^{-1}$ or $h$ into $h^{-1}$, it follows that $W_g=W_{g^{-1}}$ and $W_{h}=W_{h^{-1}}$. The fourth relation for metric compatibility can be therefore written as
\[ B_g + B_{g^{-1}}+2\left(\frac{1}{B_{h^{-1}}}-1\right)=0, \] so by changing $h\leftrightarrow h^{-1}$ and subtracting resulting equations
we get $B_h=B_{h^{-1}}$ and similarly also $B_g=B_{g^{-1}}$. Therefore the above equation reduces to $B_g+\frac{1}{B_h}=1$.
The third relation for metric compatibility is now of the form,
\begin{equation*}
\left(\frac{1}{B_g}-1\right)(B_g-1)+2W_gW_h=0,
\end{equation*}
hence by the symmetry of the second term we infer
\[ \left(\frac{1}{B_g}-1\right)(B_g-1)=\left(\frac{1}{B_h}-1\right)(B_h-1).\]
Using now $B_g+\frac{1}{B_h}=1$ it can be reduced to an algebraic equation
\[(1-B_g)^3=B_g^3,\]
which has three solutions
\[ B_g=\tfrac{1}{2}, \qquad B_g=\tfrac{1}{2}\big(1\pm {\rm i}\sqrt{3}\big).\]
Since we still have an analogue of equation \eqref{iksy_b}, $R_gB_g=\frac{1}{B_g}$, the first solution is excluded. From the second one we deduce that the function $B_g$ can take values only in the set
\[ \big\{\tfrac{1}{2}\big(1+ {\rm i}\sqrt{3}\big), \tfrac{1}{2}\big(1- {\rm i}\sqrt{3}\big)\big\},\] and if $B_g(n,m)=\frac{1}{2}\big(1\pm {\rm i}\sqrt{3}\big)$, then $B_g(n+1,m)$ has to be equal to $\frac{1}{2}\big(1\mp {\rm i}\sqrt{3}\big)$. Obviously such solutions are possible only if $N$ is even. Moreover, from $\frac{1}{B_g}+B_h=1$ we can deduce also that $B_h=B_g$, hence $R_hB_g=\frac{1}{B_g}$, so we have a similar behaviour also in the second argument. As a~result there are two possible solutions
\[ B_g(n,m)=\tfrac{1}{2}\big(1+(-1)^{n+m}{\rm i}\sqrt{3}\big),\qquad B_g(n,m)=\tfrac{1}{2}\big(1-(-1)^{n+m}{\rm i}\sqrt{3}\big),\]
and moreover the existence of such solutions requires both~$N$ and~$M$ to be even. Furthermore, in such a case we have $2W_hW_g=-1$, which determines the values of~$W$'s (and, using cotorsion-freeness, also of~$C$'s). Indeed, using the fifth relation for metric compatibility (which now is of the form $W_g(2+W_h+W_g)=0$) together with the condition $2W_hW_g=-1$ we get $W_g=-1-\sqrt{\frac{3}{2}}$ and $W_h=\sqrt{\frac{3}{2}}-1$, or with the exchanged role of indices $h$ and $g$.

Therefore at least one of $W$s needs to vanish unless both $N$ and $M$ are even when the aforementioned possibility occurs, however, if it is not the case, we shall see that it is not possible that only one of $W$ is zero.
Indeed, suppose the contrary, i.e., without loss of generality assume that only $W_{h^{-1}}=0$. First notice that the last
relation in cotorsion-freeness implies
\[ C_{g^{-1}h}-1=W_{h^{-1}}.\]
Applying the above relation (together with $R_hA_h=\frac{1}{B_{h^{-1}}}$ and the analogue of \eqref{iksy_b} which is still valid here) in the fourth condition for metric compatibility we get
\[ W_{h^{-1}}(B_{g}+B_{g^{-1}})+\left(\frac{1}{B_{h^{-1}}}-1\right)(W_h+W_{h^{-1}})=0.\]
Since $W_{h^{-1}}=0$ we get $B_{h^{-1}}=1$ or $W_h=0$. Since we had assumed only one $W$ vanishes, $B_{h^{-1}}=1$. Applying the same technique to the third condition for metric compatibility as for the fourth one, we infer
\[ \left(\frac{1}{B_g}-1\right)(B_{g^{-1}}-1)+W_g(W_h+W_{h^{-1}})=0.\]
Replacing $g$ with $h$, and using $B_{h^{-1}}=1$, $W_h\neq 0$ we end up with $W_g=-W_{g^{-1}}$. If $W_{g}\neq 0$, then also $W_{g^{-1}}$. In this case the fifth condition for metric compatibility (after using cotorsion-freeness) reduces to $W_{g^{-1}}\big(2+W_{g^{-1}}\big)=0$, so both $W_{g^{-1}}$ and $W_{g}$ (since $W_h,W_{h^{-1}}\neq 0$) are equal to $-2$. But the only possibility to satisfy $W_g=-W_{g^{-1}}$ is now $W_g=W_{g^{-1}}=0$, which is a~contradiction. Therefore, the claim is proven.

Furthermore, notice that since from $W_{h}\neq W_{h^{-1}}=0$ we were able to deduce that $W_g=W_{g^{-1}}=0$, it follows that if one $W$ vanishes then there exists at least one pair of vanishing $W$s: $W_a=W_{a^{-1}}=0$. Therefore it is not possible that exactly two of $W$ with indices in different algebras
vanish simultaneously.

Therefore even in the case of a constant metric, such that the lengths of sides are
the same in all directions, the solution is not uniquely determined by the requirement of torsion-freeness and metric-compatibility. In addition to the trivial solution with all $W$ being zero (which reduces to the case discussed in the previous subsection), there are also other possibilities, e.g., with $W_h=W_{h^{-1}}=0$ and $W_{g}=W_{g^{-1}}=-2$. For $N$ and $M$ with even parities, there are even more sophisticated solutions with alternating functions~$B$.

\section{Conclusions and overview}
In this paper we posed the question, whether it is possible to classify all linear connections over the minimal differential calculi on
the finite cyclic group that are torsion-free and compatible with a given metric,
extending the already known results and the program of systematic studies of
noncommutative Riemannian geometry (see~\cite{BM_book} and the references therein).

Surprisingly, even though the problem is nonlinear the answer is positive yet only possible for a class of metrics that are either
proportional-symmetric (left and right metrics are proportional to each other) or satisfy very special relations that are quantized.
However, only the proportional-symmetric solutions are meaningful in the sense of Riemann geometry, as only they can lead
to a norm on the space of one-forms. The resulting linear connections yield a nontrivial scalar curvature for the Riemannian
geometry of the discretized circle, which has an interesting continuous limit.

The extension of the construction of bimodule connections and compatible metrics to the products
of two discretized circles leads to a highly nontrivial set of compatibility conditions and this paper
only scratches the surface of the problem. Yet, we were able to show that for the constant metric
there exists at least one torsion-free linear connection that is compatible with it. This example shows that torsion-freeness and
metric compatibility are not so restrictive conditions as in the classical situation and even in the simplest case we can have
a plenty of non-trivial solutions.

There remain two important problems: the uniqueness of the linear connection for a class of
nonconstant metrics as well as the relation of the computed scalar curvature to the spectral
analysis through the Dirac operator~\cite{Ba15,BG19,Gl17} for discretized models, which we leave
for the forthcoming work.

\subsection*{Acknowledgements}
We would like to thank the referees for their valuable comments on the content of our manuscript and their suggestions for improving the document. PZ was supported by the Faculty of Physics, Astronomy and Applied Computer Science of the Jagiellonian University under the DSC scheme: U1U/P05/NW/03.27.

\pdfbookmark[1]{References}{ref}
\LastPageEnding

\end{document}